\documentclass[a4paper,10pt]{article}

\usepackage{fullpage,latexsym,amsthm,amsmath,color,amssymb,url,hyperref}
\usepackage{tikz}

\usetikzlibrary{arrows,decorations.pathreplacing,shapes}
\usepackage{alphabeta}

\usepackage{color}
\usepackage{soul}

\input{colordvi}

\newcommand{\mynewtheorem}[2]{
  \newaliascnt{#1}{dummy}
  \newtheorem{#1}[#1]{#2}
  \aliascntresetthe{#1}
  \expandafter\def\csname #1autorefname\endcsname{#2}
}

\newtheorem{lemma}{Lemma}
\newtheorem{claim}{Claim}

\newtheorem{observation}{Observation}

\newtheorem{definition}{Definition}
\newtheorem{theorem}{Theorem}

\newcommand{\cost}{\mathsf{cost}}

\newcommand{\ie}{\emph{i.e.}}

\definecolor{Red}{rgb}{1, 0 ,0}
\definecolor{Blue}{rgb}{0, 0 ,1}

\newcommand{\red}[1]{{\color{Red}{#1}}}
\newcommand{\blue}[1]{{\color{Blue}{#1}}}
\newcommand{\avms}{\mathbf{\sf avms}}

\newcommand{\mavms}{\mathbf{\sf mavms}}

\newcommand{\ctp}{\mathbf{\sf ctp}}

\newcommand{\A}{(\textsf{L1})}
\newcommand{\B}{(\textsf{L2})}
\newcommand{\C}{(\textsf{L3})}
\newcommand{\width}{\mathsf{width}}

\newcommand{\rspace}{\mathsf{fsp}}
\newcommand{\clear}{\mathsf{clear}}

\newcommand{\play}{\mathcal{P}}
\newcommand{\s}{\mathbf{s}}

\newcommand{\f}{\mathbf{f}}

\newcommand{\W}{\mathcal{W}}

\title{\textbf{The mixed search  game against an agile and visible fugitive is monotone\footnote{Research supported by projects DEMOGRAPH (ANR-16-CE40-0028), ESIGMA (ANR-17-CE23-0010), and the French-German Collaboration ANR/DFG Project UTMA (ANR-20-CE92-0027).}}}

\author{Guillaume Mescoff\thanks{LIRMM, Univ Montpellier, CNRS, Montpellier, France.}
\and Christophe Paul$^{\dagger}$  \and Dimitrios M. Thilikos$^{\dagger}$}

\date{}

\begin{document}

\maketitle

\begin{abstract}
\noindent We consider the mixed search game against an agile and visible fugitive.
This is the variant of the classic fugitive search game on graphs where searchers may be placed to (or removed from) the vertices  or 
 slide along edges. Moreover, the fugitive resides on the edges of the graph and can move at any time along unguarded paths. 
{The \emph{mixed search number against an agile and visible fugitive} of a graph $G$, denoted $\avms(G)$,  is the minimum number of searchers required to capture a fugitive in this graph searching variant.}
Our main result is that this graph searching variant is {\sl monotone} in the sense 
that the number of searchers required for a successful search strategy does not increase if we restrict the search strategies to those that do not permit the fugitive {to} visit an already ``clean'' edge.
This means that mixed search strategies against an agile and visible fugitive can be polynomially 
certified, and therefore that {the problem of deciding, given a graph $G$ and an integer $k,$ whether $\avms(G)\leq k$} is in {\sf NP}. Our proof is based on the introduction of the notion of
{\sl tight bramble}, that serves as an obstruction for the corresponding search parameter.
Our results imply that for a graph $G$, $\avms(G)$ is equal to the Cartesian tree product number of $G$, that is, the minimum $k$ for which $G$ is a minor of the Cartesian product of a tree and a clique on $k$ vertices. 
\end{abstract}
\medskip


%\tableofcontents
%\newpage
%---------------------------------------------------------------------------------------------------------------
%---------------------------------------------------------------------------------------------------------------
\section{Introduction}

A \emph{search game} on a graph is a two-player game opposing a fugitive and a group of searchers. The searchers win if they capture the fugitive, while the fugitive's objective is to avoid capture. The searchers and the fugitive play in turns. Search games on graphs were originally been introduced by Parsons~\cite{Parsons78Pursu} who defines the \emph{edge search number} of a graph as the minimum number of searchers required to capture a fugitive moving along the edges. To capture the fugitive, a search program is seen as a sequence of searchers' moves. A move consists in either adding a searcher on a vertex, removing a searcher from a vertex or sliding a searcher along an edge. The fugitive is allowed to move along the edges of paths that are free of searchers. The \emph{edge search number} is then the minimum number of searchers required  to  capture the fugitive (that is, the fugitive cannot escape through an unguarded path). 

During the last decades, a large number of variants of {seminal} Parsons' edge search game {have} been studied. 
{Each of}
these variants depends on the fugitive and searchers abilities. For example, the fugitive can be \emph{visible} to   the searchers or \emph{invisible}. It can also be \emph{lazy}, then it stays at its location as long it is not directly threatened by the searchers, or \emph{agile}, then it may change its location at every round. Where the fugitive resides (on edges or on vertices), whether the searchers are allowed to slide on edges or not, how the fugitive is captured are other criteria used to define search game variants (see~\cite{FominT08Anann,Nisse19Netwo} for surveys).

The theory of search games on graphs is strongly connected to the theory of width parameters and graph decompositions. It is well known that the node search number of a graph against an agile but visible fugitive corresponds to the \emph{treewidth} of that graph~\cite{SeymourT93Graph}. The same equivalence also holds for  the case of the node search number against a lazy but invisible fugitive~\cite{DendrisKT97Fugit}. The node search variant where the fugitive is agile and invisible corresponds to the \emph{pathwidth}~\cite{KirousisP86Searc,Kinnersley92Theve,EllisST94Theve}. These three variants are usually defined as node search games as the  {searchers} cannot slide along an edge and the fugitive is located on vertices {rather than on edges.} In the case of agile but visible fugitive (respectively lazy but invisible fugitive), the existence of an escape strategy for the fugitive can be derived from the existence of a \emph{bramble} of large order~\cite{SeymourT93Graph} which is a certificate of large treewidth. Likewise, in the case of agile and invisible, a large \emph{blockage} certifies a large pathwidth and allows an escape strategy~\cite{BienstockRST91Quick}. 

Interestingly, a consequence of the equivalences discussed above between search numbers and width parameters is that the corresponding games can be proved to be \emph{monotone}. Intuitively, a game is monotone if it allows to optimally search a graph without recontamination of the previously cleared part of the graph.  Monotonicity  is an important property as it implies that the problem of 
deciding the corresponding search number of a graph belongs to $\mathsf{NP}$ (see~\cite{FominT08Anann} for a discussion about monotonicity).

\paragraph{Our results.} In this paper, we consider the so-called \emph{mixed search game against an agile and visible fugitive}. In this variant, the fugitive resides on the edges and the searchers can capture the fugitive either by sliding along the edge location it resides or by occupying the two vertices incident to that edge location (see \autoref{sec_mixed_search_games} for formal definitions).
Up to our knowledge, the monotonicity issue   of a mixed search game was only proved in the case of an agile and invisible fugitive~\cite{BienstockS91Monot}. We prove that the mixed search game against an agile and visible fugitive is monotone. To that aim, we first show that the corresponding search number of a graph is equal to the \emph{Cartesian tree product number} of that graph~\cite{Harvey14ontree}, also known as the \emph{largeur d'arborescence}~\cite{ColinDeVerdiere98multi}. To certify these two parameters, we introduce the notion of \emph{loose tree-decomposition} which is a relaxation of the celebrated tree-decomposition associated to tree-width. This allows us to define the notion of \emph{tight bramble} as the min-max counterpart of  Cartesian tree product number. Then the monotonicity of the mixed search game against an agile and visible fugitive is obtained by proving that the existence of a tight bramble allows an escape strategy for the fugitive, while the existence of a losse tree-decomposition of a small width allows to derive a monotone search strategy of small cost. 
%All of our proofs can be adapted in the case of a lazy but invisible fugitive.\sed{More!}

%---------------------------------------------------------------------------------------------------------------
%---------------------------------------------------------------------------------------------------------------

%---------------------------------------------------------------------------------------------------------------
\section{Preliminairies}

%---------------------------------------------------------------------------------------------------------------
\subsection{Basic concepts}

In this paper, all graphs are finite, loopless, and without multiple  edges. 
%every graph $G$ we consider is simple. 
We denote by  $V(G)$ the set of the vertices of $G$ and by $E(G)$ the set of the edges of $G$.
We use notation $xy$ in to denote an edge $e=\{x,y\}$. Given a edge $xy\in E(G)$,
we define $G-xy=(V(G),E(G)\setminus\{xy\}$).
%respectively denote the set of vertices of $G$ and of edges of $G$. 
%An edge $e\in E(G)$ incident to vertices $x\in V(G)$ and $y\in V(G)$ is denoted by $e=\{x,y\}$.  
The graph resulting from the removal of a vertex subset $S$ of $V(G)$ is denoted by $G-S$. 
If $S=\{v\}$ is a singleton we write $G-v$ instead of $G-\{v\}$. 
We define   the subgraph of $G$ {\em  induced by} $S$ as the graph $G[S]=G-(V(G)\setminus S)$.  We say that $S$ is {\em  connected} (in $G$) if $G[S]$ is connected.
A subset $S$ of $V(G)$ is a \emph{separator} of $G$ if $G-S$ contains more connected components than $G$. We say that $S$ separates $X\subseteq V(G)$ from $Y\subseteq V(G)$ if $X$ and $Y$ are subsets of distinct connected components of $G-S$. 
%\cor{
%Let $F$ be a non-empty subset of $E(G)$. We define $V(G)_{|F}=\{x\in V(G)\mid \exists e\in F\mbox{ such that $x$ is incident to $e$}\}$ and say that $F$ is \emph{connected} if the subgraph $H=(V(G)_{|F},F)$ of $G$ is connected.
%}

Contracting an edge $e=xy$ in a graph $G$ yields the graph $G_{/e}$ where $V(G_{/e})=V(G)\setminus\{x,y\}\cup\{v_e\}$ and $E(G_{/e})=\{uv\in E(G)\mid \{u,v\}\cap\{x,y\}=\emptyset\}\cup\{v_ev\mid xv\in E(G) \mbox{ or } yv\in E(G)\}$.
A graph $H$ is a \emph{minor} of a graph $G$, which is denoted $H\preceq G$, if $H$ can be obtained from $G$ by a series of vertex or edge deletions and edge contractions. It is well known that if $H\preceq G$, then there exists a \emph{minor model} of $H$ in $G$, that is a function $\rho:V(H)\rightarrow 2^{V(G)}$ such that: 
\begin{enumerate}
\item  for every $x\in V(H)$, the set $\rho(x)$ is connected in  $G$; and  
\item for every distinct vertices $x,y\in V(H)$, $\rho(x)\cap\rho(y)=\emptyset$; and 
\item for every edge $xy\in E(H)$, there exists an edge $x'y'\in E(G)$ with $x'\in\rho(x)$ and $y'\in\rho(y)$.
\end{enumerate}

Given a set $S$,   a \emph{sequence of subsets of $S$} is defined as $\mathcal{S} = \langle S_1,\dots,S_i,\dots \rangle$ with {$i \geq 1$,}  where for every {$i\geq 1$}, $S_i \subseteq S$. The symmetric difference between two sets $A$ and $B$ is denoted by $A\ominus B$.

A \emph{pathway} $\W$ in a graph $G$ is a sequence $\W=\langle e_1, e_2, \dots, e_{\ell}\rangle$ of edges of $G$ such that if $\ell >1$, then for every $i\in[\ell-1]$, the edges $e_i$ and $e_{i+1}$ are {distinct and} incident to a common vertex (that is, $e_i\cap e_{i+1}=\{v\}$ with $v\in V(G)$). If $\W$ is a pathway starting at edge $e$ and ending at edge $e'$, we say that $\W$ is a $(e,e')$-pathway. Observe that a pathway may contains several occurence of the same edge. 

%Let $\mathcal{S}_1 = \langle S_1,\dots,S_p \rangle$ and $\mathcal{S}_2 = \langle S_1',\dots,S_{p'}' \rangle$ be two sequences of vertices of $G$ with $p,p' \in \mathbb{N}^\ast$. If $S_p = S_1'$, therefore, we denote by $ \mathcal{S}_1::\mathcal{S}_2$ the sequence $\langle S_1,\dots,S_p,S_2',\dots,S_{p'}'\rangle$. Given $S  \subseteq V(G)$ with $p = |S|$, $[S]^{-1}$ denotes a sequence $\langle S_1,S_2 \dots,S_{p+1} \rangle$ with $ S_1 = S$ and $S_{p+1} = \emptyset$ where for every $i \in [p]$, $S_{p+1} \subseteq S_p$ and $|S_p|+1 = |S_{p+1}| $. It corresponds to a sequence where we remove vertices of $S$ one by one. Also, given $S \subseteq V(G)$, $[S] = \langle S_1,\dots, S_{|S|+1} \rangle$ corresponds to a sequence where $S_1 = \emptyset$, $S_{|S|+1} = S$ and where for all $i \in [|S|]$ , $S_{i} \subseteq S_{i+1} \wedge \, |S_i| -1 =|S_{i+1}|  $. It corresponds to a sequence where we add on by one the vertices of $S$. Again, given $S_1,S_2 \subseteq V(G)$ with $S_2 \subseteq S_1$, $[S_1 \backslash S_2]$ denotes the sequences where we remove one by one vertices of $S_1$ until we obtain $S_2$.

%\newpage

%---------------------------------------------------------------------------------------------------------------
\subsection{Cartesian tree product number}

\label{sec_ctp}

Let $T$ be a tree and $k$ be a {strictly} positive integer. We use $K_{k}$ in order to denote the complete graph on $k$ vertices. We let $T^{(k)}=T\Box K_k$ denote the \emph{cartesian product} of $T$ with $K_k$, that is: $T^{(k)}$ is the graph with vertex set $\{(x,i) \mid x \in V(T),i \in [k]\}$ and with an edge between $(x,i)$ and $(y,j)$ when $x = y$, or when $xy\in E(T)$ and $i = j$ (see \autoref{fig_cartesian_product}). 

%\xtof{Do we want to define $T^{(0)}$ as the graph on $V(T)$ without any edge ? I suggest NOT.}

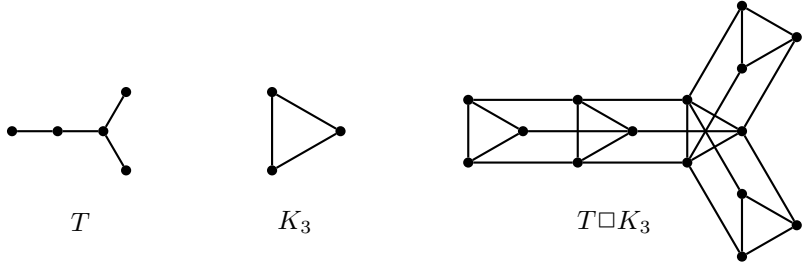
\begin{figure}[h]
\centering
\begin{center}
\begin{tikzpicture}[thick,scale=0.6]
\tikzstyle{sommet}=[circle, draw, fill=black, inner sep=0pt, minimum width=3pt]

\begin{scope}[xshift=-1.2cm]
\draw[] node[sommet] (1) at (0:0){};
\draw[] node[sommet] (2) at (60:1){};
\draw[] node[sommet] (3) at (180:1){};
\draw[] node[sommet] (4) at (300:1){};
\draw[] node[sommet] (5) at (180:2){};
\foreach \i in {2,3,4}{
\draw[-] (1) to (\i) ;
}
\draw[-] (3) to (5) ;
\node[] (K) at (-0.5,-2) {$T$};
\end{scope}

\begin{scope}[xshift=3cm]
\draw[] node[sommet] (a) at (120:1){};
\draw[] node[sommet] (b) at (240:1){};
\draw[] node[sommet] (c) at (0:1){};

\draw[-] (a) to (b) ;
\draw[-] (b) to (c) ;
\draw[-] (c) to (a) ;

\node[] (K) at (0,-2) {$K_3$};
\end{scope}

\begin{scope}[xshift=12cm]
\begin{scope}[shift=(180:4.8)]
\draw[] node[sommet] (a5) at (120:0.8){};
\draw[] node[sommet] (b5) at (240:0.8){};
\draw[] node[sommet] (c5) at (0:0.8){};
\draw[-] (a5) to (b5) ;
\draw[-] (b5) to (c5) ;
\draw[-] (c5) to (a5) ;
\end{scope}

\begin{scope}[shift=(180:2.4)]
\draw[] node[sommet] (a3) at (120:0.8){};
\draw[] node[sommet] (b3) at (240:0.8){};
\draw[] node[sommet] (c3) at (0:0.8){};
\draw[-] (a3) to (b3) ;
\draw[-] (b3) to (c3) ;
\draw[-] (c3) to (a3) ;
\end{scope}

\begin{scope}[shift=(0:0)]
\draw[] node[sommet] (a1) at (120:0.8){};
\draw[] node[sommet] (b1) at (240:0.8){};
\draw[] node[sommet] (c1) at (0:0.8){};
\draw[-] (a1) to (b1) ;
\draw[-] (b1) to (c1) ;
\draw[-] (c1) to (a1) ;
\end{scope}

\begin{scope}[shift=(60:2.4)]
\draw[] node[sommet] (a2) at (120:0.8){};
\draw[] node[sommet] (b2) at (240:0.8){};
\draw[] node[sommet] (c2) at (0:0.8){};
\draw[-] (a2) to (b2) ;
\draw[-] (b2) to (c2) ;
\draw[-] (c2) to (a2) ;
\end{scope}

\begin{scope}[shift=(300:2.4)]
\draw[] node[sommet] (a4) at (120:0.8){};
\draw[] node[sommet] (b4) at (240:0.8){};
\draw[] node[sommet] (c4) at (0:0.8){};
\draw[-] (a4) to (b4) ;
\draw[-] (b4) to (c4) ;
\draw[-] (c4) to (a4) ;
\end{scope}

\foreach \i in {2,3,4}{
\draw[-] (a1) to (a\i) ;
\draw[-] (b1) to (b\i) ;
\draw[-] (c1) to (c\i) ;
}
\draw[-] (a5) to (a3) ;
\draw[-] (b5) to (b3) ;
\draw[-] (c5) to (c3) ;

\node[] (H) at (-2,-2) {$T\Box K_3$};

\end{scope}
\end{tikzpicture}
\end{center}
\caption{The cartesian product  $T^{(3)}=T\Box K_3$ of a tree $T$ and the clique $K_3$.
}
\label{fig_cartesian_product}
\end{figure}

\begin{definition} \cite{Harvey14ontree,ColinDeVerdiere98multi}
The \emph{cartesian tree product number} of a graph $G$ is
$$\ctp(G)=\min\{k\in\mathbb{N}\mid G\preceq T^{(k)}\}.$$
\end{definition}

{Observe that for any tree $T$, as $T=T\Box K_1$, $\ctp(T)=1$.}
In~\cite{Harvey14ontree}, it is shown that the cartesian tree product number of a graph is equal to the \emph{largeur d'arborescence} of that graph, a parameter introduced by Yves Colin De Verdière in~\cite{ColinDeVerdiere98multi}.

%---------------------------------------------------------------------------------------------------------------
\subsection{Loose tree-decomposition}

The concept of \emph{loose tree-decomposition} is a relaxation of the celebrated \emph{tree-decomposition}~\cite{Halin76S-func,RobsertsonS84GMIII} associated to the tree-width of a graph. 

\begin{definition}[Loose tree-decomposition] \label{def_ltd}
Let $G=(V,E)$ be a graph. A \emph{loose tree-decomposition} is a  pair ${\cal D}=(T,\chi)$ such that $T$ is a tree and $\chi:V(T)\to 2^{V(G)}$ satisfying the following properties:
\begin{enumerate}
\item[\emph{\A}] for every vertex $x \in V(G)$, the set $\{t\in V(T)\mid x\in \chi(t)\}$ induces a non-empty connected subgraph, say $T_{x}$, of $T$. We call $T_{x}$ the {\em trace} of $x$ in ${\cal D}$;
\item[\emph{\B}] for every edge $e = xy \in E(G)$, there exists a tree-edge $f=\{t_{1},t_{2}\}\in E(T)$ such that $e \in E(G[\chi(t_1)\cup\chi(t_2)])$;
\item[\emph{\C}] for every tree-edge $f=\{t_{1},t_{2}\}$, $\big|E(G[\chi(t_1)\cup\chi(t_2)]) \setminus  \big(E(G[\chi(t_1)])\cup E(G[\chi(t_2)])\big)\big|\leq 1$.
\end{enumerate}
We refer to the vertices of $T$ as the {\em nodes} and to their images via $χ$ as the {\em bags} of the loose tree decomposition ${\cal D}=(T,χ)$. 
The {\em width} of a loose tree-decomposition $\mathcal{D}=(T,\chi)$ of a graph $G$ 
%as the size of its largest bag, that 
is defined as ${\sf width}(\mathcal{D},G)=\max\{|\chi(t)|\mid t\in V(T)\}$. %\xtof{propagate the $-1$.}%The \emph{loose treewidth} of a graph $G$ is:
%$$\ltw(G)=\min\{{\sf width}(\mathcal{D},G)\mid \mathcal{D} \mbox{ is a loose tree-decomposition of } G\}.$$
\end{definition}

\autoref{fig_ltd} provides an example of a loose tree-decomposition. Hereafter, if $\mathcal{D}=(T,\chi)$ is a loose tree-decomposition of $G$, and $xy\in E(G[\chi(t_1)\cup\chi(t_2)]) \setminus  (E(G[\chi(t_1)])\cup E(G[\chi(t_2)]))$ for some adjacent nodes  $t_1$ and $t_2$ (see condition \C{} above), then $xy$ is called a \emph{marginal edge}. We now examine important properties of loose tree-decompositions that will be used in further proofs.

\begin{figure}[ht]
\begin{center}
\begin{tikzpicture}[thick,scale=1]
\tikzstyle{sommet}=[circle, draw, fill=black, inner sep=0pt, minimum width=3pt]
          
%3-sun 
\begin{scope}[yshift=-0.25cm,scale=1.2]
\draw node[sommet] (a) at (90:1){};
\draw node[sommet] (d) at (210:1){};
\draw node[sommet] (f) at (330:1){};
\draw node[sommet] (e) at (270:0.5){};
\draw node[sommet] (b) at (150:0.5){};
\draw node[sommet] (c) at (30:0.5){};

\node[] at (90:1.2) {$a$};
\node[] at (150:0.7) {$b$};
\node[] at (30:0.7) {$c$};
\node[] at (210:1.2) {$d$};
\node[] at (270:0.7) {$e$};
\node[] at (330:1.2) {$f$};

\draw (90:1) -- (210:1);
\draw[line width=2pt] (90:1) -- (30:0.5);
\draw[] (30:0.5) -- (330:1);
\draw (210:1) -- (330:1);
\draw (150:0.5) -- (30:0.5);
\draw (150:0.5) -- (270:0.5);
\draw (270:0.5) -- (30:0.5);

%\node[] at (270:1.5) {$3$-sun};
\end{scope}

\begin{scope}[xshift=5cm,scale=0.9]
\node[] (a) at (0,0){};
\node[] (b) at (2,0){};
\node[] (c) at (4,1){};
\node[] (d) at (4,-1){};
\node[] (e) at (6,1){};
\node[] (f) at (6,-1){};

\draw (a) -- (b) ;
\draw (b) -- (c) ;
\draw (b) -- (d) ;
\draw (c) -- (e) ;
\draw (d) -- (f) ;

\draw[fill=white] (a) circle (0.5) ;
\node[] at (a) {$\mathbf{a},b$};

\draw[fill=white] (b) circle (0.6) ;
\node[] at (b) {$b,\mathbf{c},e$};

\draw[fill=white] (c) circle (0.5) ;
\node[] at (c) {$b,e$};

\draw[fill=white] (d) circle (0.5) ;
\node[] at (d) {$c,e$};

\draw[fill=white] (e) circle (0.5) ;
\node[] at (e) {$e,d$};

\draw[fill=white] (f) circle (0.5) ;
\node[] at (f) {$e,f$};

\end{scope}

\end{tikzpicture}
\end{center}
\caption{\label{fig_ltd} A loose tree-decomposition (on the right) $\mathcal{D}$ of the $3$-sun $G$ (graph on the left). Observe that the edge $ac$ is a marginal edge as no  {bag} contains the two vertices $a$ and $c$, they are however contained in two {adjacent} bags. Note that $bd$ and $cf$ are also marginal edges. We have $\width(\mathcal{D},G)=3$. Notice that if $G'=G-\{a\},$ then $G'$ has a loose tree decomposition of width 2. This decomposition is obtained from $\mathcal{D}$  by removing the bags $\{a,b\}$ and $\{b,c,e\}$ and by making adjacent the nodes corresponding to the bags $\{b,e\}$ and $\{c,e\}$. Notice that in this new loose tree decomposition, the edge $bc$ is a marginal edge.}
\end{figure}
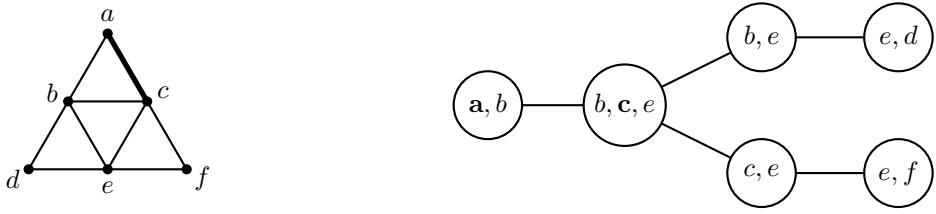

A loose tree-decomposition $\mathcal{D}=(T,\chi)$ of width $k\in\mathbb{N}$ is \emph{full} if for every node $t\in V(T)$, $|\chi(t)|=k$ and if for every pair of adjacent nodes $t$ and $t'$ in $T$, $\chi(t)\ominus\chi(t')=\{x,x'\}$ with $x\in \chi(t)$ and $x'\in\chi(t')$. The next lemma shows that every loose tree-decomposition can be turned into a full one of the same width.

\begin{lemma} \label{lem_full_ltd}
Given a loose tree-decomposition $\mathcal{D}=(T,\chi)$ of a graph $G$, one can compute a full loose tree-decomposition $\mathcal{D}'=(T',\chi')$ of $G$ such that ${\width}(\mathcal{D},G)={\width}(\mathcal{D}',G)$.
\end{lemma}
\begin{proof}
Suppose that ${\sf width}(\mathcal{D},G)=k$. We can assume that in $\mathcal{D}$, if two nodes $t,t'$ of $T$ are adjacent, then neither $\chi(t)\subseteq \chi(t')$ nor $\chi(t')\subseteq \chi(t)$. Let $t\in V(T)$ be a node such that $|\chi(t)|=k$. Suppose there exists a node $t'\in V(T)$ adjacent to $t$ such that $|\chi(t')|<k$. As there is at most one marginal edge $e$ between $t$ and $t'$, it is possible to add vertices from $\chi(t)$ that are not incident to $e$ to complete $\chi(t')$ to a bag of size $k$. Observe that completing every bag of $T$ that way does not violate the conditions of \autoref{def_ltd}. We can process all the nodes of $\mathcal{D}$ to obtain a loose tree-decomposition whose bags all have size $k$. Let us assume it is the case for $\mathcal{D}$.

Suppose that $t_1$ and $t_2$ are adjacent nodes in $T$ such that $|\chi(t_1)\cap\chi(t_2)|<k-1$. So there exists $x_1,y_1 \in \chi(t_1) \setminus \chi(t_2)$ and some $x_2,y_2 \in \chi(t_2) \setminus \chi(t_1)$. Since there exists at most one marginal edge between $t_1$ and $t_2$, we can suppose without loss of generality that $x_1x_2\notin E(G)$. Then we can remove from  $T$ the edge $t_1t_2$ and insert a new node $t$ adjacent to $t_1$ and $t_2$ such that $\chi(t)=(\chi(t_1)\cup\chi(t_2))\setminus \{x_1,x_2\}$. Again, one can observe that this transformation does not violate the conditions of \autoref{def_ltd}. 
It follows that processing {$T$ as long as it contains a tree-edge on which this transformation applies, yields} to transform $\mathcal{D}$ into a full loose tree-decomposition $\mathcal{D}$ without increasing the width.
\end{proof}

\begin{lemma} \label{lem_node_sep}
Let $\mathcal{D}=(T,\chi)$ be a full loose tree-decomposition of a graph $G$. For every node $t$ that is not a leaf of $T$, $\chi(t)$ is a separator of $G$.
\end{lemma}
\begin{proof}
As $t$ is not a leaf, $T$ contains two nodes $t_1$ and $t_2$ adjacent to $t$. As $\mathcal{D}$ is full, there exist $x_1\in \chi(t_1)\setminus\chi(t)$ and $x_2\in\chi(t_2)\setminus \chi(t)$. Let $P$ be a path in $G$ between $x_1$ and $x_2$. Let us recall that every edge of $G$ is either a marginal edge between two neighboring nodes of $T$ or incident to two vertices belonging to the bag of some node of $T$. This implies that the set of nodes $\{t\in V(T)\mid t\in T_x,x\in P\}$ induces a connected subtree of $T$. As $t_1$ and $t_2$ belong to that set, so does $t$. It follows that $P$ intersects $\chi(t)$ and thereby $\chi(t)$ is a separator of $G$.
\end{proof}

\begin{lemma} \label{lem_tree_edge_sep}
Let $\mathcal{D}=(T,\chi)$ be a full loose tree-decomposition of a graph $G$ and let $\{t_1,t_2\}$ be an edge of $T$. If $\chi(t_1)\ominus\chi(t_2)=\{x_1,x_2\}\notin E(G)$, then $\chi(t_1)\cap\chi(t_2)$ is a separator of $G$.
\end{lemma}
\begin{proof}
Suppose that $x_1\in\chi(t_1)\setminus\chi(t_2)$ and $x_2\in\chi(t_2)\setminus\chi(t_1)$. We prove that $\chi(t_1)\cap\chi(t_2)$ is a $(x_1,x_2)$-separator. Let $P$ be a path between $x_1$ and $x_2$ in $G$. We let $T_1$ denote the largest subtree of $T$ containing $t_1$ but not $t_2$. Consider the vertex set  $V_1=\{x\in V(G)\mid \exists t\in T_1, x\in \chi(t)\}$. 
{Let $x\in V_1$ and $y\notin V_1$ be the two adjacent vertices of $P$ such that every vertex $z$ of $P$ between $x_1$ and $x$  belongs to $V_1$.} As $x_2\notin V_1$, the vertices $x$ and $y$ are well defined {($x$ may be equal to $x_1$)}. By \autoref{def_ltd}, every edge is either a marginal edge between two adjacent nodes of $T$ or is contained in the bag of some node of $T$. By assumption, there is no marginal edge between $t_1$ and $t_2$. As {$x\in V_1$ and $y\notin V_1$}, we have that $x\in \chi(t_1)\cap\chi(t_2)$, and thereby $\chi(t_1)\cap\chi(t_2)$ intersects $P$.
\end{proof}

%---------------------------------------------------------------------------------------------------------------
\subsection{Tight brambles}
\label{sec_tight_bramble}

{In order to define a notion of obstacle to the existence of a loose tree-decomposition of small width, we} now adapt the well-known definition of bramble that is used in the context of tree-decomposition and treewidth~\cite{SeymourT93Graph}.
%The proof of the main theorem of this subsection strictly follows the line of the proof of~\cite{BellenbaumD02Twosh}.
Let $G$ be a graph. 
%\cor{Two connected subsets $F_1$ and $F_2$ of $E(G)$ are \emph{tightly touching} if $E(G)$ contains two distinct edges $x_1x_2$ and $y_1y_2$ such that $x_1,y_1\in V(G)_{|F_1}$ and $x_2,y_2\in V(G)_{|F_2}$. Observe that these two edges may be incident to a common vertex. If $x_2=y_1$, then we have $V(G)_{|F_1}\cap V(G)_{|F_2}\neq\emptyset$.} For simplicity, in this paper we use {\em touching} as a shortcut of the term {\em tightly touching}.
Two subsets $S_1$ and $S_2$ of $V(G)$ are \emph{tightly touching} if either $S_1\cap S_2\neq\emptyset$ or $E(G)$ contains two distinct edges $x_1x_2$ and $y_1y_2$ such that $x_1,y_1\in S_1$ and $x_2,y_2\in S_2$ {(the edges $x_1x_2$ and $y_1y_2$ may share one vertex but not two, as graphs are considered to be simple).}
For simplicity, in this paper we use {\em touching} as a shortcut of the term {\em tightly touching}.

\begin{definition}[Tight bramble] \label{def_tight_bramble}
Let $G$ be a graph. A set $\mathcal{B}\subseteq 2^{V(G)}$ of pairwise touching connected subsets of vertices, each of size at least two, is a \emph{tight bramble} of $G$. A set $S\subseteq V(G)$ is a \emph{cover} of $\mathcal{B}$ if for every set $B\in\mathcal{B}$, $S\cap B\neq\emptyset$. The \emph{order} of the bramble $\mathcal{B}$ is the smallest size of a cover of $\mathcal{B}$. 
\end{definition}

%\cor{
%\begin{definition}[Tight bramble] \label{def_tight_bramble}
%Let $G$ be a graph. A set $\mathcal{B}\subseteq 2^{E(G)}$ of pairwise touching connected subsets of $E(G)$ is a \emph{tight bramble} of $G$. A set $S\subseteq V(G)$ is a \emph{cover} of $\mathcal{B}$ if for every set $B\in\mathcal{B}$, $S\cap V(G)_{|B}\neq\emptyset$. The \emph{order} of the bramble $\mathcal{B}$ is the smallest size of a cover of $\mathcal{B}$. 
%\end{definition}
%}

%\xtof{add the  remark that we can assume that a bramble element contains at least one edge.}

\begin{figure}[ht]
\begin{center}
\begin{tikzpicture}[thick,scale=1]
\tikzstyle{sommet}=[circle, draw, fill=black, inner sep=0pt, minimum width=3pt]
          
%3-sun 
\begin{scope}[]

\draw node[sommet] (a) at (90:1){};
\draw node[sommet] (b) at (180:1){};
\draw node[sommet] (c) at (270:1){};
\draw node[sommet] (d) at (0:1){};

\node[] at (90:1.2) {$a$};
\node[] at (180:1.2) {$b$};
\node[] at (270:1.2) {$c$};
\node[] at (0:1.2) {$d$};

\draw (a) -- (b) -- (c) -- (d) -- (a) -- (c) ;
\draw (b) -- (d) ;
%\node[] at (270:1.5) {$3$-sun};
\end{scope}

\begin{scope}[shift=(0:-6)]

\draw node[sommet] (a) at (0:1){};
\draw node[sommet] (b) at (60:1){};
\draw node[sommet] (c) at (120:1){};
\draw node[sommet] (d) at (180:1){};
\draw node[sommet] (e) at (240:1){};
\draw node[sommet] (f) at (300:1){};

\node[] at (0:1.2) {$a$};
\node[] at (60:1.2) {$b$};
\node[] at (120:1.2) {$c$};
\node[] at (180:1.2) {$d$};
\node[] at (240:1.2) {$e$};
\node[] at (300:1.2) {$f$};

\draw (a) -- (b) -- (c) -- (d) -- (e) -- (f) -- (a) ;
\end{scope}

\end{tikzpicture}
\end{center}
\caption{\label{fig_tight_bramble} 
On the left, we observe that $\mathcal{B}=\big\{\{a,b\},\{b,c\},\{c,d,e,f,a\}\big\}$ is a tight bramble of order two. On the right,
$\mathcal{B}_1=\big\{\{x,y\}\mid xy\in E(G) \big\}$ is a tight bramble of order $4$ of the $K_4$. Observe that $\mathcal{B}_2=\big\{ \{a\},\{b,c\},\{c,d\},\{b,c\} \big\}$ is a set of pairwise touching connected subsets of $V(K_4)$. However, according to \autoref{def_tight_bramble}, as it contains the singleton set $\{a\}$, is it not a tight bramble. See  \autoref{obs_edgeless} and the related discussion.}

\end{figure}
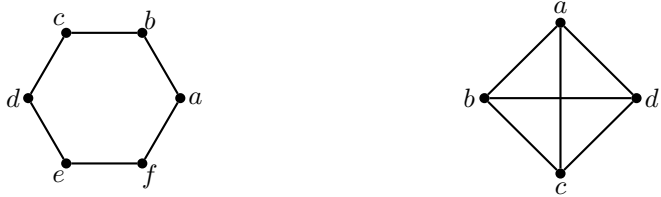

Let us discuss \autoref{def_tight_bramble}. Observe first that, if $G$ is an edge-less graph, then its unique tight bramble is $\mathcal{B}=\emptyset$ which has order $0$. In the rest of the paper, we will not consider edge-less graphs as searching a fugitive (located on edges) in such a graph would create  degenerate situations. The next observation provides a characterization of trees by means of tight brambles.

\begin{observation} \label{obs_bramble_tree}
Let $G$ be a connected graph. The maximum order of a bramble of $G$ is $1$ if and only if $G$ is a tree.
\end{observation}
\begin{proof}
We observe that if $G$ is a tree, a tight bramble consists of set of pairwise intersecting subtrees of $G$. It follows by he Helly property that the maximum order of a bramble of a tree is $1$. Suppose that $G$ is not a tree. Then $G$ contains a cycle $C$ of at least $3$ edges. Let $f_1,f_2,\dots f_t$, $t\geq 3$, be the edges the edges of $C$ such that for $1\leq i< t$, $f_i\cap f_{i+1}\neq\emptyset$ and $f_1\cap f_t\neq\emptyset$. Then observe that $\mathcal{B}=\big\{\{f_1\},\{f_2\},C\setminus\{f_1,f_2\}\big\}$ is a tight bramble of order $2$ (see \autoref{fig_tight_bramble}).
\end{proof}

Finally, if we discard edge-less graphs, then we can observe that the size-two constraint on the elements of a tight bramble can be relaxed without changing the value of the maximum order of a tight bramble of a given graph. %Suppose that $\mathcal{B}\subset 2^{V(G)}$ is a set of pairwise touching connected subsets of $V(G)$ containing the singleton set $\{x\}$, for some $x\in V(G)$. First notice that, as $G$ is simple, every other element $S\in \mathcal{B}$ has size at least two. We claim

\begin{observation} \label{obs_edgeless}
Let $\mathcal{B}\subset 2^{V(G)}$ be a set of pairwise touching connected subsets of $V(G)$ containing the singleton set $\{x\}$, for some $x\in V(G)$. If the order of $\mathcal{B}$ is $k$, then there exists a tight bramble $\mathcal{B}'$ of $G$ of order $k$.
\end{observation}
\begin{proof}
To construct $\mathcal{B}'$, we proceed as follows. For every $S\in\mathcal{B}$ such that $S\neq\{x\}$, let $y_S$ and $z_S$ be two vertices of $S$ such that $xy_S\in E$ and $xz_S\in S$. Then we set $\mathcal{B}'=\big(\mathcal{B}\setminus\{x\}\big)\bigcup \big\{\{xy_S\},\{xz_S\}\mid S\in\mathcal{B},S\neq\{x\}\big\}$. In other words, $\mathcal{B}'$ is obtained by replacing $\{x\}$ by every pair of edges $\{xy_S\}$ and $\{xz_S\}$. Clearly, the order of $\mathcal{B}'$ is at most $k$. To see this, consider a cover $X$ of $\mathcal{B}$. Observe that $x\in X$ and thereby $X$ is a cover of $\mathcal{B}'$ as well. Now assume $\mathcal{B}'$ has a cover $X'$ not containing $X$. Let us consider $S\in \mathcal{B}\cap\mathcal{B}'$. Then $S$ contains two vertices of $S$, namely $y_S$ and $z_S$. It is easy to see that $X'\setminus \{y_S\}\cup\{x\}$ is also a cover of $\mathcal{B}'$. It follows that $\mathcal{B}'$ has a cover of minimum order that contains $x$ and that is thereby also a cover of $\mathcal{B}$.
\end{proof}

%
%\xtof{@Dimitrios: observe that without the size-two constraint, if $G$ is the edge-less graph, then $\mathcal{B}=\{x\}$, for $x\in V(G)$, has order $1$. Which is not what was expected!}
%

\begin{lemma}\label{bramble_separation}
%Let $\mathcal{F}$ be a set of connected subset of vertices of a graph $G$ and let $S_1$ and $S_2$ be two covers of $\mathcal{F}$. If $S\subsetneq V(G)$ separates $S_1$ from $S_2$ in $G$, then $S$ is a cover of $\mathcal{F}$.
Let $\mathcal{B}$ be a tight bramble of a graph $G$ and let $S_1$ and $S_2$ be two covers of $\mathcal{B}$. If $S\subsetneq V(G)$ separates $S_1$ from $S_2$ in $G$, then $S$ is a cover of $\mathcal{B}$.
\end{lemma}
\begin{proof}
Let $\mathcal{B}$ be a tight bramble of a graph $G$ and $S_1,S_2$ two covers of $\mathcal{B}$. Let $S \subsetneq V(G)$ be a separator of $S_1$ and $S_2$. Consider a set $B \in \mathcal{B}$ of the tight bramble. Since $S_1$ and $S_2$ cover $B$ and since $B$ is connected, there exists a path $P$ whose internal vertices belong to $B$ and whose endpoints belong to $S_1$ and $S_2$. Thus, since $S$ separates $S_1$ and $S_2$, there exists a vertex $x \in P$ such that $x \in S$ and thus, $S$ covers $B$. This holds for every $B \in \mathcal{B}$, thus $S$ is a cover of $\mathcal{B}$. 
\end{proof}

%---------------------------------------------------------------------------------------------------------------
\subsection{Mixed search games}
\label{sec_mixed_search_games}

In this paper, we deal with mixed  search games on a graph $G$ introduced by Bienstock and Seymour~\cite{BienstockS91Monot}.
The opponents are a team of searchers and a fugitive. We
assume that both players have full knowledge of the graph and  the position of their opponent in the graph. Moreover, the fugitive is {\em agile} in the sense that it may move at any moment though unguarded pathways.
At each round, the searchers occupy a subset of vertices, called \emph{searchers' position}, while the fugitive is located on an edge, called \emph{fugitive location}. 
A \emph{play}  is a (finite or  infinite) sequence alternating between searchers' positions and fugitive locations, that is a sequence 
%and such that the first fugitive location is $e_1$, that is:
${\cal P}=\langle S_0, e_1, S_1, \dots, e_{\ell}, S_{\ell},\dots \rangle$
with $S_0=\emptyset$ and where  for $i\geq 1$, $S_i\subseteq V(G)$ is a searchers' position and $e_i\in E(G)$ is a fugitive location. 
%Hereafter, the sets $S_i$'s, for $i\geq 0$, are called the \emph{searchers' positions}. 
In case a play  is finite,  it always ends with the special symbol $\star$, replacing a fugitive location and indicating the fact that the fugitive has been captured. 
%How a fugitive is captured by the searchers will be discussed later.
%The game stops when the fugitive has been captured by the searchers. The capture condition depends on the game variant and will be discussed later.
The {\em cost} of a play  ${\cal P}$ is the maximum size of a searchers' position $S_{i}$ in it and is denoted by $\cost({\cal P})$.

\paragraph{The search strategy.} 
%Let $\play_{G,e_1}=\langle S_0, e_1, S_1, \dots, e_{\ell}, S_{\ell},\dots \rangle$ be the play starting at $e_1$ generated by \blue{the program $(\s_G,e_{1},\f_G)$}. 
A {\em searchers' move} is a pair $(S,S')\in \big(2^{V(G)}\big)^2$ indicating a transition from the searchers' position $S$ to searcher's position $S'$.  
%
%and $S_i=\si_G(S_{i-1})$ is called a \emph{searchers' move}.
In a mixed search game, a  searchers' move is {\em legitimate} if $|S\ominus S'|\in\{1,2\}$ and, moreover,  
if $|S\ominus S'|=2$ then $S\ominus S'$ is an edge of $G$. This {allows} three types of moves:
\begin{itemize}
\item {[{\sl Placement of a searcher}]:} $S'=S\cup\{x\}$ for some $x\in V(G)\setminus S$. This type of move consists in placing a new searcher on vertex $x$.
\item {[{\sl Removal of a searcher}]:} $S=S'\setminus \{x\}$ for some $x\in S$. This type of move consists in  removing a searcher from vertex $x$.
\item {[{\sl Sliding of a searcher along an edge}]:} $S\ominus S'=\{x,y\}$ and $xy\in E(G)$.  This type of move corresponds to sliding the searcher initially positioned at vertex $x$, along the edge {$xy$}, towards eventually occupying vertex $y$.
\end{itemize}
As the searchers know the exact {fugitive location}, they can use this information to determine their next positions. Therefore, we may define a  {\em (mixed) search strategy} on $G$ as a function ${\bf s}_{G}\in\big(2^{V(G)}\big)^{\big(2^{V(G)}\times E(G)\big)}$
such that, for every $(S,e)\in 2^{V(G)}\times E(G)$, the pair $\big(S,{\bf s}_{G}(S,e)\big)$ is a legitimate searchers' move. We denote by ${\cal S}_{G}$  the set of all search strategies on $G$.
Every legitimate searchers' move $(S,S')$ 
immediately clears a set of edges defined as follows:
%$$C(S,S'):=\big\{\{x,v\}\mid x\in S\big\}\cap E(G),$$
$$\clear_G(S,S'):=
\begin{cases}
\big\{xy\mid x\in S\big\}\cap E(G), & \mbox{if }S'\setminus S=\{y\}\\
\emptyset & \mbox{if }S'\setminus S=\emptyset.\end{cases}
$$

\paragraph{The fugitive strategy.}

%We let $\mathcal{F}_G$ denote the set of fugitive strategies.
%, we observe that $\mathcal{F}_G\subseteq \big(E^{\star}(G)\big)^{\big(2^{V(G)}\times 2^{V(G)}\times E(G)\big)}$.
To formally define  the fugitive strategy we need to introduce some concepts. Let us consider a legitimate  searchers' move $(S,S')$. 

A {\em pathway}  $\W=\langle f_1,f_2,\dots f_t\rangle$ of $G$, with $t\geq 2$, is \emph{$(S,S')$-avoiding} if the following conditions are satisfied:
\begin{enumerate}
\item for every $j\in[2,t]$, $f_{j-1}\cap f_j\notin S\cap S'$, and

\item if $S'\setminus\{y\}= S\setminus\{x\}$  (that is, $(S,S')$ is a sliding move from $x$ to $y$ along the edge $xy$, see \autoref{fig_avoiding_pathway}), and $xy\in \mathcal{W}$, then 
 \begin{itemize}
 \item $xy=f_1$ or $xy=f_t$;
 \item if $xy=f_1=f_t$, then $t>2$.
 \end{itemize}
\end{enumerate}

%
%\red{
%A {\em pathway}  $\W=\langle f_1,f_2,\dots f_t\rangle$ of $G$, with $t\geq 2$, is \emph{$(S,S')$-avoiding} if the following conditions are satisfied:
%\begin{enumerate}
%\item $S'=S\cup\{y\}$ (place in $y$, see \autoref{fig_avoiding_pathway}): for every $j\in[2,t]$, $f_{j-1}\cap f_j\notin S\cap S'$.
%
%\item $S'=S\setminus\{x\}$ (remove from $x$):  for every $j\in[2,t]$, $f_{j-1}\cap f_j\notin S\cap S'$.
%
%\item $S'\setminus\{y\}= S\setminus\{x\}$ (slide from $x$ to $y$ along the edge $xy$, see \autoref{fig_avoiding_pathway}):
% for every $j\in[2,t]$, $f_{j-1}\cap f_j\notin S\cap S'$ and if $xy\in\mathcal{W}$, then 
% \begin{itemize}
% \item $xy=f_1$ or $xy=f_t$;
% \item if $xy=f_1=f_t$, then $t>2$.
% \end{itemize}
%\end{enumerate}
%
%\noindent [\textbf{@Dimitrios:}] please check this new definition.
%}

%A {\em pathway}  
%$\W=\langle f_1,f_2,\dots f_t\rangle$ of $G$, with $t\geq 2$, is \emph{$(S,S')$-avoiding} if the following conditions are satisfied:
%\begin{itemize}
%\item If $f_1\in  \mbox{$Clear_G(S,S')$}$ (see \autoref{fig_avoiding_pathway}), then
%\begin{itemize}
%\item[(1)] {$f_1\cap f_2= S'\setminus S$} (here we see edges as sets of 2 vertices);
%\item[(2)] $f_2\notin \mbox{${S\cup S' \choose 2}\cap E(G)$}$; and 
%\item[(3)] for every $j\in[3,t]$, $f_{j-1}\cap f_j\notin S'$.
%\end{itemize}
%\item Otherwise, 
%for every $j\in[2,t]$, $f_{j-1}\cap f_j\notin S'$.
%\end{itemize}

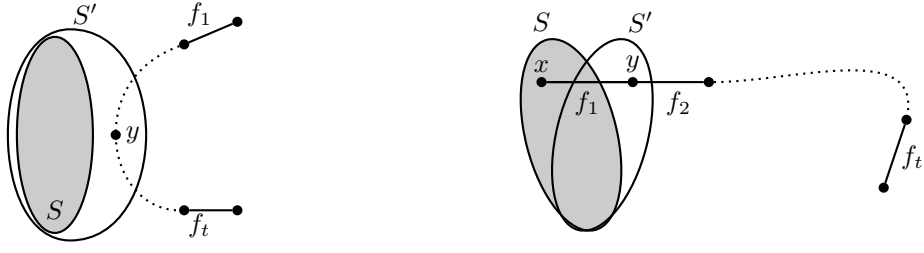
\begin{figure}[ht]
\begin{center}
\begin{tikzpicture}[thick,scale=1]
\tikzstyle{sommet}=[circle, draw, fill=black, inner sep=0pt, minimum width=3pt]

%-------- placement ----------
\begin{scope}
\draw[rotate=90,fill=gray!40] (0,0.2) circle (1.3 and 0.5);
\node[] at (-0.2,-1) {$S$};
%\draw[rotate=90] (0,0) circle (1.3 and 1);

\coordinate[] (a) at (90:1.4);
\coordinate[] (b) at (-90:1.4);
\coordinate[] (c) at (0:1);
\draw[-,>=latex] (a) edge[out=0,in=90] (c) ;%edge[out=90,in=0] (b);
\draw[-,>=latex] (c) edge[out=-90,in=0] (b) ;
\draw[-,>=latex] (b) edge[out=180,in=180] (a) ;
\node[] at (0.2,1.6) {$S'$};

\coordinate[] (x) at (0.6,0);
\coordinate[] (v) at (1.5,1);
\coordinate[] (w) at (1.5,-1);

\draw node[sommet] (u) at (0.6,0){};
\node[right] at (u) {$y$};

\draw node[sommet] (v) at (1.5,1.2){};
\draw node[sommet] (u) at (2.2,1.5){};
\draw node[sommet] (w) at (1.5,-1){};
\draw node[sommet] (z) at (2.2,-1){};

\node[] at (1.7,1.6) {$f_1$};
\node[] at (1.7,-1.2) {$f_t$};

\draw (u) -- (v);
\draw[dotted,>=latex] (v) edge[out=200,in=90] (x) ;
\draw[dotted,>=latex] (x) edge[out=-90,in=180] (w) ;
\draw (w) -- (z) ;

\end{scope}

%-------- sliding ----------
\begin{scope}[shift=(0:7)]
\draw[rotate=104,fill=gray!40] (0.1,0.4) circle (1.3 and 0.6);
\draw[rotate=90-14] (0,0) circle (1.3 and 0.6);

\draw node[sommet] (u) at (-0.8,0.7){};
\node[above] at (u) {$x$};
\draw node[sommet] (v) at (0.4,0.7){};
\node[above] at (v) {$y$};
\draw (u) -- (v);
\node[below] at (-0.2,0.7) {$f_1$};

\draw node[sommet] (a) at (1.4,0.7){};
\draw node[sommet] (b) at (4,0.2){};
\draw node[sommet] (c) at (3.7,-0.7){};

\draw (u) -- (a);
\node[below] at (1,0.7) {$f_2$};
\draw (b) -- (c);
\node[right] at (3.8,-0.3) {$f_t$};

\draw[dotted,>=latex] (a) edge[out=0,in=80] (b) ;
%\draw[dotted]  (a) .. controls (2.5,-0.5)  and  (3.4,1) .. (b);

\node[] at (-0.8,1.5) {$S$};
\node[] at (0.5,1.5) {$S'$};

\end{scope}

\end{tikzpicture}
\end{center}
\vspace{-3mm}
\caption{\label{fig_avoiding_pathway} Two legitimate moves $(S,S')$ and $(S,S')$-avoiding pathway $\mathcal{W}=\langle f_1,f_2,\dots f_t\rangle$ are depicted. On the left, we have $S'=S\cup\{x\}$, and $\mathcal{W}$ may go through the vertex $x$. On the right, $S\ominus S'=\{x,y\}$, i.e., a searcher slides along the edge {$f_1=xy$} from $x$ to $y$. Observe that, if $t>2$, then the edge $f_t$ may be the edge $f_1$.
}
%$(S,S')$  is  a legitimate searchers' move  where $S\ominus S'=\{x,y\}$, i.e., a searcher slides along the edge {$f_1=xy$}, implying that $f_1\in Clear_G(S,S')$. The figure depicts  an $(S,S')$-avoiding pathway $\mathcal{W}=\langle f_1,f_2,\dots f_t\rangle$. Observe that, if $t>2$, then the edge $f_t$ may be the edge $f_1$.
\end{figure}
Given a legitimate  searchers' move $(S,S')$ and an edge $e$, the subset of edges of $E(G)$ that are \emph{accessible} from $e$ through a $(S,S')$-avoiding pathway of $G$ is defined as
$$A_G(S,e,S'):=\left\{e'\in E(G)\setminus \mbox{${S' \choose 2}$} \mid  \mbox{there is {an} $(S,S')$-avoiding $(e,e')$-pathway}\right\}.$$
%\xtof{Has the notion of $(S_{i-1},S_i)$-avoiding to be defined depending on the considered game ? Check the edge search.}\\ 
We can now define the {\em fugitive space} so as to contain  the edges  where the fugitive   initially residing at $e$ may move after the {searchers'} move $(S,S')$, that is 
$$\rspace_{G}(S,e,S'):= \left(\{e\}\setminus \clear_G(S,S')\right) \cup A_G(S,e,S').$$
%\red{
%Observe that if the edge $e$ is cleared by the searchers' move  $(S,S')$ (i.e. $e\in C_i$), then $e'\in A_G(S_{i-1},S_i,e)$. This means that if $e_{i-1}$ is going to be cleared by the searchers moving from $S_{i-1}$ to $S_i$, then the fugitive has to escape along a $(S_{i-1}, S_i)$-avoiding pathway.
%}
%
%\paragraph{The search program.}
%\red{As for the searchers, the next edge location of the fugitive is determined by a function, $\f_G$, called fugitive strategy. At each step  the fugitive strategy takes as input the consecutive searchers' positions $S$ and $S'$, the current fugitive location $e$ and either returns a new fugitive location $e'$ or (in the case the fugitive is captured) the special value $\star$. }
Given a graph $G$, we set $E^{\star}(G)=E(G)\cup\{\star\}$
and we define a \emph{fugitive strategy} on $G$ as a  pair $(e_1,\f_{G})\in E(G)\times \big(E^{\star}(G)\big)^{\big(2^{V(G)}\times E(G)\times 2^{V(G)}\big)}$ such that, for the function $\f_{G}$, if $\rspace_{G}(S,e,S')\neq\emptyset$, then $\f_G(S,e,S')\in \rspace_{G}(S,e,S')$,  otherwise $\f_G(S,e,S')=\star$.  
We denote by  ${\cal F}_{G}$  the set if all fugitive strategies on $G$.

%
%
%It follows from the above discussions that a play $\play_{G,e_1}$ starting at some edge $e_1\in E(G)$ is generated by a pair of functions $\play_G=(\s_G,e_{1},\f_G)$, hereafter called \emph{program}, where $\s_G$ is the \emph{search strategy} and $\f_G$ is the \emph{fugitive strategy}, such that:
%$$\play_G=(\s_G,e_{1},\f_G)\in (\mathcal{S}_G,\mathcal{F}_G)\subseteq \big(2^{V(G)}\big)^{\big(2^{V(G)}\times E(G)\big)}\times \big(E^{\star}(G)\big)^{\big(2^{V(G)}\times 2^{V(G)}\times E(G)\big)};$$ 
%and such that:
%\begin{itemize}
%\item for $i\geq 1$, the searchers' move from $S_{i-1}$ to $S_i=\s_G(S_{i-1},e_i)$ is legitimate;
%\item for $i>1$, if $\rspace_{G}(S_{i-1},S_i,e_{i-1})\neq\emptyset$, then $\f_G(S_{i-1},e_{i-1},S_i)\in \rspace_{G}(S_{i-1},S_i,e_{i-1})$, and otherwise $\f_G(S_{i-1},e_{i-1},S_i)=\star$. 
%\end{itemize}

%---------------------------------------------------------------------------------------
\paragraph{Winning and monotone search strategies}
 
A search program (on a graph $G$) is a pair $\big({\bf s}_{G},(e_{1},{\bf f}_{G})\big)\in{\cal S}_{G}\times {\cal F}_{G}$. 
{Every search} program $\big({\bf s}_{G},(e_{1},{\bf f}_{G})\big)$ {\em generates} {a play, denoted}
\begin{eqnarray}
\play({\bf s}_{G},e_{1},{\bf f}_{G}) & := & \langle S_0, e_1, S_1, \dots, e_{\ell}, S_{\ell},\dots \rangle
\label{kiolp}
\end{eqnarray}
where for each $i\geq 1$, $S_{i}={\bf s}_{G}(S_{i-1},e_i)$ and $e_{i+1}={\bf f}_{G}(S_{i-1},e_{i},S_{i})$.
%
%The \emph{cost}, denoted by $\cost({\bf s}_{G},e_{1},{\bf f}_{G})$, of a search program $({\bf s}_{G},e_{1},{\bf f}_{G})$ is the cost of $\play({\bf s}_{G},e_{1},{\bf f}_{G})$.
%
% 
%$$\cost({\bf s}_{G},e_{1},{\bf f}_{G}):=\max\big\{|S|\mid \exists e_1\in E(G) \mbox{: $S$ is a searchers' position of~} \play({\bf s}_{G},e_{1},{\bf f}_{G},e_{1})\big\}$$
The {search} program $(\s_G,e_{1},\f_G)$ is \emph{monotone} if, in \eqref{kiolp}, for every $i\geq 1$, the edge $e_i$ has not been cleared at {any step prior to $i$}, that is for every $j\leq i$, $e_i\notin \clear_G(S_{j-1},S_j)$.
A search strategy $\s_G\in\mathcal{S}_G$ is \emph{monotone} if for every fugitive strategy $(e_1,\f_G)\in \mathcal{F}_G$, the program $\big(\s_G,(e_{1},\f_G)\big)$ is monotone.

Let ${\bf s}_{G}$ be a search strategy. 
The \emph{cost}, denoted by $\cost(\s_G),$ of  $\s_G$ is the maximum cost of $\play({\bf s}_{G},e_{1},{\bf f}_{G})$, over all $(e_{1},{\bf f}_{G})\in{\cal F}_{G}$.
%$\cost(\s_G):=  \max\big\{\cost(\s_G,e_1,\f_G) \mid (e_{1},\f_G)\in \mathcal{F}_G\big\}.$
%
%A program $({\bf s}_{G},e_{1},{\bf f}_{G})$  is \emph{searcher-winning} if  .
%, i.e.,  $$\play_{G,e_1}=\langle S_0,e_1,\dots S_{i-1},e_i,S_i,\star\rangle.$$
Also, $\s_G$ is \emph{winning} if for every fugitive strategy $(e_1,\f_G)\in \mathcal{F}_G$, the play $\play({\bf s}_{G},e_{1},{\bf f}_{G})$  is finite.

%
%
%
%isible search game $(\mathcal{S}_G,\mathcal{F}_G)$ (\xtof{Shall we rather restrict a game to the set of programs ?}) against an agile fugitive is \emph{monotone} if for every winning strategy $\s_G\in \mathcal{S}_G$, there exists a winning and monotone strategy $\mathbf{\sf s'}_G\in \mathcal{S}_G$ such that $\cost(\mathbf{\sf s'}_G)\leq \cost(\s_G)$. 

We define the \emph{mixed search number (against an agile and visible fugitive)} of $G$ as:
$$\avms{}(G)= \min \big\{\cost(\s_G)\mid \s_G \mbox{~is a winning search strategy in~} \mathcal{S}_G  \big\}. $$
And if we restrict the search strategies to the monotone ones, we define:
$$\mavms{}(G)=  \min \big\{\cost(\s_G)\mid \s_G \mbox{~is a monotone winning search strategy in~} \mathcal{S}_G  \big\}. $$

%---------------------------------------------------------------------------------------------------------------
%---------------------------------------------------------------------------------------------------------------

\section{Monotonicity of the mixed search game against an agile and visible fugitive}

This section is devoted to the main result of this paper, that is \autoref{th_min_max} {below.} It provides a min-max characterization of graphs with Cartesian tree product number at most $k$ in terms of tight bramble obstacles. Moreover, it shows that the Cartesian tree product number corresponds to the number of searchers required to capture an agile and visible fugitive in a mixed search game. As a byproduct, it proves that the mixed search game against an agile and visible fugitive is monotone.
This results can be seen as the ``visible counterpart'' of the result of Bienstock and Seymour \cite{BienstockS91Monot} where the monotonicity is proved for the case of an agile and invisible fugitive. Also, it can be seen as the ``mixed counterpart'' of the result of   Seymour and Thomas \cite{SeymourT93Graph} where  the monotonicity is proved for the {\sl node} search game against  an agile and visible fugitive.

%
%, answering an open problem~\cite{???}. \xtof{What is the good reference for that ?}\sed{We have to mention this in a different way!}

\begin{theorem} \label{th_min_max}
Let $G$ be a graph and $k\geq 1$ be an integer. %$k\in\mathbb{N}$. 
Then the following conditions are equivalent:
\begin{enumerate}
\item $G$ has a loose tree-decomposition of width $k$;
\item $G$ is a minor of $T^{(k)}=T\Box K_k$ (\ie~$\ctp(G)\leq k$);
\item every tight bramble $\mathcal{B}$ of $G$ has order at most $k$;
%$G$ does not contain a tight bramble of order $k+1$;
\item the mixed search number against an agile and visible fugitive is at most $k$ (\ie{} $\avms(G)\leq k$);
\item the monotone mixed search number against an agile and visible fugitive is at most $k$ (\ie{} $\mavms(G)\leq k$).
\end{enumerate}
\end{theorem}

The proof of \autoref{th_min_max} is as follows. First, \autoref{th_ctp_width} proves that $(1\Leftrightarrow 2)$. Then \autoref{th_bramble} establishes $(3\Rightarrow 2)$. In turn  \autoref{th_avms_la} shows $(2\Rightarrow 5)$. The fact that $(5\Rightarrow 4)$ is trivial. Finally, we can conclude, using \autoref{th_escape}, that $(4\Rightarrow 3)$.

%As a byproduct of the proofs of \autoref{th_min_max}, we obtain that the mixed search game against an agile and visible fugitive is monotone.
%
%\begin{theorem} \label{th_monotone}
%For every graph $G$, we have $\mavms(G)=\avms(G)$.
%\end{theorem} 

%---------------------------------------------------------------------------------------------------------------
\subsection{The Cartesian tree product number and loose-tree decomposition}

%\xtof{Explain more. What about $\la(G)$ ?} 
\autoref{th_ctp_width} establishes that the minimum width of a loose tree-decomposition is equal to the  \emph{Cartesian tree product number} of a graph ($(1\Leftrightarrow 2)$ from~\autoref{th_min_max}). %~\cite{Harvey14ontree}.
We first show that every minor of $T^{(k)}=T\Box K_k$ admits a loose tree-decomposition of width $k$. To that aim, we prove that the width of a loose tree-decomposition of a graph is a parameter that is closed under the minor relation (\autoref{lem_ctw_minor_close}). Then, we build for $T^{(k)}$ a loose tree-decomposition of width at most $k$. To prove the reverse direction, we show that given a loose tree-decomposition $\mathcal{D}$ of width $k$ of $G$, one can complete the graph $G$ into a graph $H$ which is contained in $T^{(k)}$ as a minor. The graph $H$ is obtained by adding all possible missing edge between pair of vertices of the same bag of $\mathcal{D}$ and all possible missing marginal edges of $\mathcal{D}$.

\begin{lemma}\label{lem_ctw_minor_close}
Let $G$ and $H$ be two graphs such that $H\preceq G$. If $G$ has a loose tree-decomposition of width $k$, then also does $H$.
% then $\ltw(H)\le \ltw(G)$.
\end{lemma}
\begin{proof}
%It is clear from the definition of a loose tree-decomposition that if $H$ is obtained by removing an edge or a vertex from a graph $G$, then $\ltw(H)\le \ltw(G)$.
Let $\mathcal{D}=(T,\chi)$ be a loose tree-decomposition of $G$. Suppose that $H$ is obtained by removing an edge or a vertex from a graph $G$. It is clear  
from \autoref{def_ltd} that restricting the bags of $\mathcal{D}$ to the vertices of $H$ yields a loose tree-decomposition of $H$  and does not increase the width.
So suppose that $H=G_{/e}$ from some edge $e=xy\in E(G)$. Let $v_e$ be the new vertex resulting from the contraction of $e$. We define $\mathcal{D}'=(T,\chi')$ such that for every $t\in V(T)$, if $x,y\notin \chi(t)$, then $\chi'(t)=\chi(t)$, otherwise $\chi'(t)=\chi(t)\setminus\{x,y\}\cup\{v_e\}$.
We prove that $\mathcal{D}'$ is a loose tree decomposition of $H$. 
\begin{itemize}
\item First observe that by condition \B{}, in $\mathcal{D}$ the trace $T_x$ of $x$ and the trace $T_y$ of $y$ either intersect or are joined by a tree edge. It follows that in $\mathcal{D}'$, the trace $T_{v_e}$ is connected. So condition {\A{}} holds for $\mathcal{D}'$. 
\item By construction of $\mathcal{D}'$, for every node $t\in V(T)$ such that $x\in\chi(t)$ or $y\in\chi(t)$, we have $v_e\in\chi'(t)$. This implies that $\mathcal{D}'$ satisfies condition \B{}. 
\item Suppose that in $\mathcal{D}'$, the edge $uv\in E(H)$ is a marginal edge for the nodes $t_1$ and $t_2$ of $T$. This means that $u\in\chi'(t_1)\setminus\chi'(t_2)$ and $v\in\chi'(t_2)\setminus\chi'(t_1)$. By construction, if $u$ and $v$ both belong to $V(G)$, then $uv$ is also a marginal edge of $G$ for the nodes $t_1$ and $t_2$ in $\mathcal{D}$. So assume,   without loss of generality that {say} $u=v_e$. As $v_e$ results from the contraction of {the} edge $xy\in E(G)$, $G$ contains the edge $xv$ or $yv$. Suppose that $xv\in E$. By construction of $\mathcal{D}'$, we have $x\in\chi(t_1)\setminus\chi(t_2)$. It follows that in $\mathcal{D}$, $xv$ is a marginal edge {between} the nodes $t_1$ and $t_2$ of $T$. So we proved that every marginal edge in $\mathcal{D}'$ corresponds to a marginal edge in $\mathcal{D}$. This implies that if for some tree-edge $t_1t_2$ of $T$, $\mathcal{D}'$ does not satisfies condition \C{}, neither does $\mathcal{D}$. %\sed{This last sentence is confusing.}
%Suppose that in $\mathcal{D}'$, there are two adjacent nodes $t_1$ and $t_2$ of $T$ such that $E(H[\chi'(t_1)\cup\chi'(t_2)]) \setminus  (E(G[\chi'(t_1)])\cup E(H[\chi'(t_2)]))$ contains two edges. Then at least one of these edges, say $e'$, is incident to $v_e$, as otherwise $\mathcal{D}$ would not satisfies condition $(3$). Suppose that $e'=v_ev$ for some vertex $v\in\chi'(t_2)$, that is $v\in\chi'(t_2)\setminus\chi'(t_1)$ and $v_e\in\chi'(t_1)\setminus\chi'(t_2)$. Suppose also that $e'$ replace in $H$ the edge $xv\in E(G)$ (the other cases are symmetric). Observe that $v\neq x$ and $v\neq y$ implies that $v\in\chi(t_2)\setminus\chi(t_1)$. Moreover we also have that $x\in\chi(t_1)\setminus\chi(t_2)$. It follows that in $\mathcal{D}$ the edge $xv\in E(H[\chi(t_1)\cup\chi(t_2)]) \setminus  (E(G[\chi(t_1)])\cup E(H[\chi(t_2)]))$. So if $\mathcal{D}'$ does not satisfies condition $(3)$, neither does $\mathcal{D}$. 
\end{itemize}
Finally, we observe that, by construction, ${\sf width}(\mathcal{D},H)\leq {\sf width}(\mathcal{D},G)$, concluding the proof.
\end{proof}

\begin{lemma} \label{th_ctp_width}
For every graph $G$, we have
$$\ctp(G)=\min\big\{\width(\mathcal{D},G)\mid \mathcal{D} \mbox{ is a loose tree-decomposition of } G  \big\}.$$
\end{lemma}
\begin{proof}
We first prove that $\min\big\{\width(\mathcal{D},G)\mid \mathcal{D} \mbox{ is a loose tree-decomposition of } G  \big\}\le \ctp(G)$.  By \autoref{lem_ctw_minor_close}, it is enough to prove the statement for $T^{(k)}$, where $k=\ctp(G)$. Recall that, by construction of $T^{(k)}$, we have that {$V(T^{(k)})=\big\{x_{(t,i)}\mid t\in V(T), i\in [k]\big\}$} and {${E(T^{(k)})}=\big\{x_{(t,i)}x_{(t,j)}\mid t\in V(T), i\neq j\big\}\cup \big\{x_{(t,i)}x_{(t',i)}\mid t,t'\in V(T), t\neq t',i\in[k]\big\}$}. We build a loose tree-decomposition $\mathcal{D}=(T',\chi)$ of $T^{(k)}$ as follows:
\begin{itemize}
\item For the sake of the construction of $T'$, we fix an arbitrary node $r$ of $T$ as the root of $T$, implying that for every node $t\in V(T)$ distinct from $r$, there is a uniquely defined parent node $p(t)$. 
The tree $T'$ is obtained from $T$ by every edge $k-1$ times.
%The tree $T'$ is obtained from $T$ by replacing every node by a path of length $k$ and linking these paths by respecting the parent relation induced by $r$. 
More formally, we define (see \autoref{fig_ctw_ctp}): 
$$V(T')=\big\{v_{(r,k)}\big\}\cup\big\{v_{(t,j)}\mid t\in V(T),t\neq r, j\in [k]\big\}, \mbox{ and}$$ 
$$E(T')=\big\{v_{(t,j)}v_{(t,j+1)}\mid t\in V(T),t\neq r, j\in[k-1] \big\}\cup \big\{v_{(p(t),k)}v_{(t,1)}\mid t\in V(T), t\neq r \big\}.$$
\item We set $\chi(v_{(r,k)})=\big\{x_{(r,i)}\mid i\in [k] \big\}$ and for every node $v_{(t,j)}\in V(T')$ such that $t\neq r$, we set 
$$\chi(v_{(t,j)})=\big\{x_{(t,j')}\mid 1\le j'\le j\big\}\cup \big\{x_{(p(t),j')}\mid j< j'\le k\big\}.$$
\end{itemize}

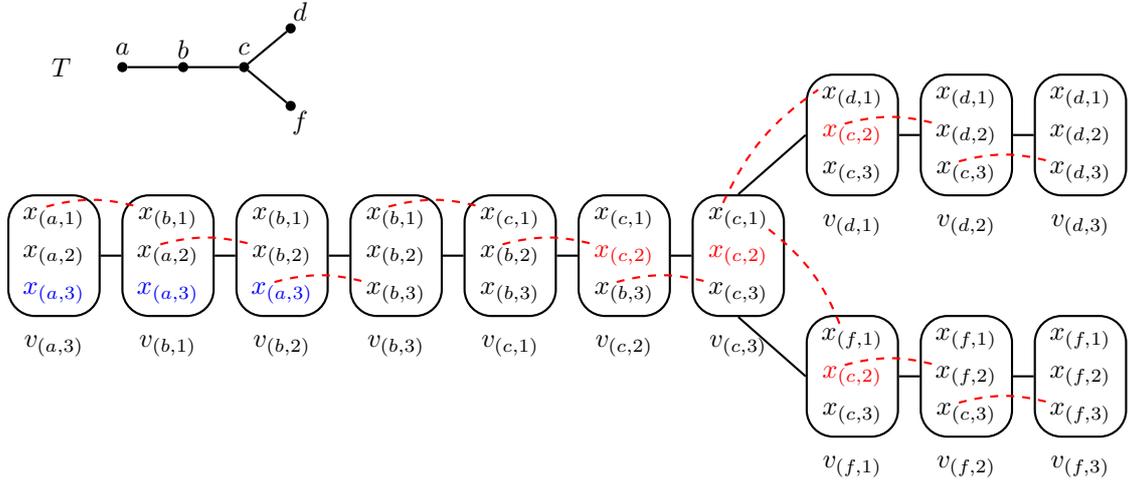
\begin{figure}[ht]
\begin{center}
\begin{tikzpicture}[thick,scale=1]
\tikzstyle{sommet}=[circle, draw, fill=black, inner sep=0pt, minimum width=3pt]
\tikzstyle{bag}=[minimum width=1.2cm,minimum height=1.6cm, rectangle,rounded corners=10pt,draw, fill=white]

\begin{scope}[xshift=-0.5cm,yshift=2.5cm,scale=0.8]
\draw[] node[sommet] (1) at (-2,0){};
\draw[] node[] () at (-2,0.3){$a$};
\draw[] node[sommet] (2) at (-1,0){};
\draw[] node[] () at (-1,0.3){$b$};
\draw[] node[sommet] (3) at (0,0){};
\draw[] node[] () at (0,0.3){$c$};
\draw[] node[sommet] (4) at (40:1){};
\draw[] node[] () at (45:1.3){$d$};
\draw[] node[sommet] (5) at (320:1){};
\draw[] node[] () at (315:1.3){$f$};

\foreach \i in {2,4,5}{
\draw[-] (3) to (\i) ;
}
\draw[-] (1) to (2) ;

\draw[] node[] () at (-3,0){$T$};

\end{scope}

\foreach \i/\j in {-3/1,-1.5/2,0/3,1.5/4,3/5,4.5/6,6/7}{
\node[bag] (N\j) at (\i,0) {};
}
\draw[-] (N1.east) to (N2.west) ;
\draw[-] (N2.east) to (N3.west) ;
\draw[-] (N3.east) to (N4.west) ;
\draw[-] (N4.east) to (N5.west) ;
\draw[-] (N5.east) to (N6.west) ;
\draw[-] (N6.east) to (N7.west) ;

\foreach \i/\j in {7.5/8,9/9,10.5/10}{
\node[bag] (N\j) at (\i,1.6) {};w
}
\draw[-] (N7.north) to (N8.west) ;
\draw[-] (N8.east) to (N9.west) ;
\draw[-] (N9.east) to (N10.west) ;

\foreach \i/\j in {7.5/11,9/12,10.5/13}{
\node[bag] (N\j) at (\i,-1.6) {};
}
\draw[-] (N7.south) to (N11.west) ;
\draw[-] (N11.east) to (N12.west) ;
\draw[-] (N12.east) to (N13.west) ;

\draw[] node[] () at (-3,0.5){$x_{(a,1)}$};
\draw[] node[] () at (-3,0){$x_{(a,2)}$};
\draw[] node[] () at (-3,-0.5){\blue{$x_{(a,3)}$}};
\draw[] node[] () at (-3,-1.2){{$v_{(a,3)}$}};

\draw[] node[] () at (-1.5,0.5){$x_{(b,1)}$};
\draw[] node[] () at (-1.5,0){$x_{(a,2)}$};
\draw[] node[] () at (-1.5,-0.5){\blue{$x_{(a,3)}$}};
\draw[] node[] () at (-1.5,-1.2){{$v_{(b,1)}$}};

\draw[] node[] () at (0,0.5){$x_{(b,1)}$};
\draw[] node[] () at (0,0){$x_{(b,2)}$};
\draw[] node[] () at (0,-0.5){\blue{$x_{(a,3)}$}};
\draw[] node[] () at (0,-1.2){{$v_{(b,2)}$}};

\draw[] node[] () at (1.5,0.5){$x_{(b,1)}$};
\draw[] node[] () at (1.5,0){$x_{(b,2)}$};
\draw[] node[] () at (1.5,-0.5){$x_{(b,3)}$};
\draw[] node[] () at (1.5,-1.2){{$v_{(b,3)}$}};

\draw[] node[] () at (3,0.5){$x_{(c,1)}$};
\draw[] node[] () at (3,0){$x_{(b,2)}$};
\draw[] node[] () at (3,-0.5){{$x_{(b,3)}$}};
\draw[] node[] () at (3,-1.2){{$v_{(c,1)}$}};

\draw[] node[] () at (4.5,0.5){$x_{(c,1)}$};
\draw[] node[] () at (4.5,0){\red{$x_{(c,2)}$}};
\draw[] node[] () at (4.5,-0.5){{$x_{(b,3)}$}};
\draw[] node[] () at (4.5,-1.2){{$v_{(c,2)}$}};

\draw[] node[] () at (6,0.5){$x_{(c,1)}$};
\draw[] node[] () at (6,0){\red{$x_{(c,2)}$}};
\draw[] node[] () at (6,-0.5){$x_{(c,3)}$};
\draw[] node[] () at (6,-1.2){{$v_{(c,3)}$}};

\draw[] node[] () at (7.5,2.1){{$x_{(d,1)}$}};
\draw[] node[] () at (7.5,1.6){\red{$x_{(c,2)}$}};
\draw[] node[] () at (7.5,1.1){{$x_{(c,3)}$}};
\draw[] node[] () at (7.5,0.4){{$v_{(d,1)}$}};

\draw[] node[] () at (9,2.1){$x_{(d,1)}$};
\draw[] node[] () at (9,1.6){$x_{(d,2)}$};
\draw[] node[] () at (9,1.1){$x_{(c,3)}$};
\draw[] node[] () at (9,0.4){{$v_{(d,2)}$}};

\draw[] node[] () at (10.5,2.1){{$x_{(d,1)}$}};
\draw[] node[] () at (10.5,1.6){$x_{(d,2)}$};
\draw[] node[] () at (10.5,1.1){$x_{(d,3)}$};
\draw[] node[] () at (10.5,0.4){{$v_{(d,3)}$}};

\draw[] node[] () at (7.5,-1.1){{$x_{(f,1)}$}};
\draw[] node[] () at (7.5,-1.6){\red{$x_{(c,2)}$}};
\draw[] node[] () at (7.5,-2.1){{$x_{(c,3)}$}};
\draw[] node[] () at (7.5,-2.8){{$v_{(f,1)}$}};

\draw[] node[] () at (9,-1.1){$x_{(f,1)}$};
\draw[] node[] () at (9,-1.6){$x_{(f,2)}$};
\draw[] node[] () at (9,-2.1){$x_{(c,3)}$};
\draw[] node[] () at (9,-2.8){{$v_{(f,2)}$}};

\draw[] node[] () at (10.5,-1.1){{$x_{(f,1)}$}};
\draw[] node[] () at (10.5,-1.6){$x_{(f,2)}$};
\draw[] node[] () at (10.5,-2.1){$x_{(f,3)}$};
\draw[] node[] () at (10.5,-2.8){{$v_{(f,3)}$}};

\draw[red,dashed] (-3.1,0.65) to[bend left=15] (-1.9,0.65){}; 
\draw[red,dashed] (-1.6,0.15) to[bend left=15] (-0.4,0.15){}; 
\draw[red,dashed] (-0.1,-0.35) to[bend left=15] (1.1,-0.35){}; 

\draw[red,dashed] (1.4,0.65) to[bend left=15] (2.6,0.65){}; 
\draw[red,dashed] (2.9,0.15) to[bend left=15] (4.1,0.15){}; 
\draw[red,dashed] (4.4,-0.35) to[bend left=15] (5.6,-0.35){}; 

\draw[red,dashed] (5.8,0.7) to[bend left=15] (7.05,2.2){}; 
\draw[red,dashed] (7.4,1.75) to[bend left=15] (8.6,1.75){}; 
\draw[red,dashed] (8.9,1.25) to[bend left=15] (10.1,1.25){}; 

\draw[red,dashed] (6.4,0.35) to[bend left=15] (7.35,-0.95){}; 
\draw[red,dashed] (7.4,-1.45) to[bend left=15] (8.6,-1.45){}; 
\draw[red,dashed] (8.9,-1.95) to[bend left=15] (10.1,-1.95){}; 
  
\end{tikzpicture}
\end{center}
\caption{\label{fig_ctw_ctp} A loose tree-decomposition $\mathcal{D}=(T',\chi)$ of $T^{(3)}=T\Box K_3$. To build $\mathcal{D}$, we set node $a$ as the root of $T$. In \blue{blue}, we have the trace $T'_{x_{(a,3)}}$ and in red the trace  $T'_{x_{(c,2)}}$. The dashed \red{red} edges are the marginal edges of $\mathcal{D}$.}
\end{figure}

One can observe that by construction of $\mathcal{D}$, the trace $T_{x_{(t,i)}}$ of every vertex is connected. Therefore condition  \A{}  of \autoref{def_ltd} holds. We observe that $T^{(k)}$ contains two types of edges. For an edge $e\in E(T^{(k)})$, either $e=x_{(t,i)}x_{(t,i')}$ for some node $t\in V(T)$ and some distinct integers $i,i'\in[k]$, or $e=x_{(t,i)}x_{(p(t),i)}$ where $t\in V(T)$ is not the root node and $i\in[k]$. In the former case, $x_{(t,i)}$ and $x_{(t,i')}$ both belong to the bag $\chi(v_{(t,k)})$. In the later case, we observe that: 1) if $i=1$, then $x_{(t,i)}\in\chi(v_{(t,1)})$ and $x_{(p(t),i)}\in\chi(v_{(p(t),k)})$; 2) otherwise $x_{(t,i)}\in\chi(v_{(t,i)})$ and $x_{(p(t),i)}\in\chi(v_{(t,i-1)})$. It follows that condition  \B{}  of \autoref{def_ltd} holds. Finally, observe that only the edges of $T^{(k)}$ of the second type are not covered by a single bag. These are the marginal edges and by construction, there is at most one marginal edge between two adjacent nodes in $\mathcal{D}$. It follows that  condition \C{} of \autoref{def_ltd} also holds and thereby $\mathcal{D}$ is a loose tree-decomposition of $T^{(k)}$. Finally, notice that ${\sf width}(\mathcal{D},T^{(k)})=k$, proving that $\min\big\{\width(\mathcal{D},T^{(k)})\mid \mathcal{D} \mbox{ is a loose tree-decomposition of } T^{(k)}  \big\}\le \ctp(T^{(k)})$.

\medskip
Let us now prove that $\ctp(G)\le \min\big\{\width(\mathcal{D},G)\mid \mathcal{D} \mbox{ is a loose tree-decomposition of } G  \big\}$. By \autoref{lem_full_ltd}, we can assume that $\mathcal{D}=(T,\chi)$ is a full loose tree-decomposition of $G$ such that ${\sf width}(\mathcal{D},G)=k$. We prove that $G\preceq T^{(k)}$. Based on $\mathcal{D}$, we construct a graph $H$ by adding edges to $G$ as follows: 
\[
\begin{array}{cl}
E(H)= & \big\{xy\mid \exists t\in V(T), \{x,y\}\subseteq\chi(t)\big\} \bigcup\\
&  \big\{xy \mid \exists t_1t_2\in E(T), x\in\chi(t_1)\setminus\chi(t_2),y\in\chi(t_2)\setminus\chi(t_1) \big\}
\end{array}
\]

That is the edge set of $H$ contains all possible edges between vertices belonging to a common bag of $\mathcal{D}$ and all possible marginal edges between adjacent bags. Clearly $E(H)$ contains $E(G)$. So, {let us} prove that $H\preceq T^{(k)}$. To that aim, we iteratively build a minor-model $\rho$ of $H$ in $T^{(k)}$ with the following algorithm. 

Initially, we set $\rho(x)=\emptyset$ for every vertex $x\in V(H)$. We guarantee that, at each step of the algorithm, if $x_{(t,i)}\in \rho(x)$ and $x_{(t,j)}\in \rho(x)$, then $i=j$. 
We pick an arbitrary node $t\in V(T)$. For every $x\in \chi(t)$, we set $\rho(x)=\{x_{(t,i)}\}$ (with $i\in[k]$) in a way that, for $y\in\chi(t)$, if $x\neq y$, $\rho(x)\neq \rho(y)$. Set $S=\{t\}$ as the set of processed nodes of $T$ (a node $t$ is processed when for every vertex $x\in\chi(t)$, $\rho(x)$ is defined). Assume there exists an unprocessed node $t\notin S$.  We  pick a node $t'\notin S$ having an adjacent node $t\in S$. Suppose that $x\in\chi(t)\cap\chi(t')$. Then $\rho(x)$ contains a vertex $x_{(t,i)}$ for some $i\in[k]$. We add $x_{(t',i)}$ to $\rho(x)$. Otherwise, suppose that $x\in\chi(t')\setminus\chi(t)$. As $|\chi(t)\cap\chi(t')|=k-1$, there is a unique vertex $y\in\chi(t)\setminus\chi(t')$. Let $j\in[k]$ be such that $x_{(t,j)}\in\rho(y)$. Then we add $x_{(t',j)}$ to $\rho(x)$.

First observe that by construction, for every $x\in V(H)$, the set $\rho(x)$ is a connected subset of $T^{(k)}$.  Let us argue that for every $xy\in E(H),$ $T^{(k)}$ contains an edge between a vertex of $\rho(x)$ and a vertex of $\rho(y)$. As explained in the description of $H$, $E(H)$ contains two type of edges $xy$: either $V(T)$ contains a node $t$ such that $xy\subseteq \chi(t)$ or $xy$ is a marginal edge, that is $E(T)$ contains a tree-edge $t_1,t_2$ such that $x\in\chi(t_1)\setminus\chi(t_2)$ and $y\in\chi(t_2)\setminus\chi(t_1)$. In the former case, by the construction of $\rho(x)$, there exist distinct $i,i'\in[k]$ such that $x_{(t,i)}\in\rho(x)$ and $x_{(t,i')}\in\rho(y)$. Observe that by definition of $T^{(k)}$, $x_{(t,i)}x_{(t,i')}\in E(T^{(k)})$. In the latter case, there exist in $T$ two adjacent nodes $t_1$ and $t_2$ such that $x\in\chi(t_1)\setminus\chi(t_2)$, $y\in\chi(t_2)\setminus\chi(t_1)$. Again, by the construction of $\rho(x)$ and $\rho(y)$, there exists a unique $i\in [k]$ such that 
$x_{(t,i)}\in\rho(x)$ and $x_{(t',i)}\in\rho(y)$. Observe that $x_{(t,i)}x_{(t',i)}\in E(T^{(k)})$. It follows that $\rho$ certifies that $H$ is a minor of $T^{(k)}$. As $G$ is a subgraph of $H$, $G$ is a minor of $T^{(k)}$ and thereby $\ctp(G)\leq k$.
\end{proof}

%---------------------------------------------------------------------------------------------------------------
\subsection{Tight brambles certify large Cartesian tree product number}

In this subsection, we show (\autoref{th_bramble}) that the existence of tight bramble of order $k$ in a graph $G$ certifies the fact that the Cartesian tree product number of $G$ is at most $k$ ($(3\Rightarrow 2)$ of \autoref{th_min_max}). Our proof  follows the lines of the proof of Bellenbaum and Diestel~\cite{BellenbaumD02Twosh} where it is shown  that, in the context of treewidth, a bramble of large order is an obstacle to a small treewidth.
We first prove a technical lemma showing that extending in a connected way the trace $T_x$ of a vertex $x$ in a loose tree-decomposition $\mathcal{D}=(T,\chi)$ of a graph $G$ yields another loose tree-decomposition of $G$.

\begin{lemma} \label{lem_td_extension}
Let $\mathcal{D}=(T,\chi)$ be a loose tree-decomposition of a graph $G$. Let $\chi':V(T)\rightarrow 2^{V(G)}$ be such that for every $x\in V(G)$, $\big\{t\in V(T)\mid x\in \chi'(t)\big\}$ induces a connected subtree $T'_x$ of $T$ and the trace $T_x$ is a subtree of $T'_x$, then $\mathcal{D}'=(T,\chi')$ is a loose tree-decomposition of $G$.
\end{lemma}
\begin{proof}
Let us show that the conditions (\textsf{L1}), (\textsf{L2}) and (\textsf{L3}) of \autoref{def_ltd} are satisfied by $\mathcal{D}'$. Clearly, (\textsf{L1}) is verified by definition of $T'_x$.
Suppose $xy$ is an edge in $E(G)$. As $\mathcal{D}$ is a loose tree-decomposition, there exists a tree-edge $f=t_1t_2\in E(T)$ such that $xy\in E\big(G[\chi(t_1)\cup\chi(t_2)]\big)$. As for every vertex $x$, $T_x$ is a subtree of $T'_x$, {we have} $xy\in E\big(G[\chi'(t_1)\cup\chi'(t_2)]\big)$. It follows that (\textsf{L2}) holds.
{Finally, observe that as, by constuction of $\mathcal{D}'$, for every vertex $x$, the trace $T_x$ is a subtree of $T'_x$, every marginal edge in $\mathcal{D}'$ is a marginal edge in $\mathcal{D}$,} implying that $\mathcal{D}'$ satisfies (\textsf{L3}).
%{Let $t_1$ and $t_2$ be two adjacent nodes in $T$. Suppose that there exist a marginal edge $xy$ such that $x\in\chi'(t_1)\setminus\chi'(t_2)$ and $y\in\chi'(t_2)\setminus\chi'(t_1)$. This implies that $T_x$ contains $t_1$ but not $t_2$ and $T'_y$ contains $t_2$ but not $t_1$. Since $T_x$ is a subtree of $T'_x$ and $T_y$ is a subtree of $T'_y$, we also have that $x\in\chi(t_1)\setminus\chi(t_2)$ and $y\in\chi(t_2)\setminus\chi(t_1)$ as otherwise $\mathcal{D}$ would not satisfy (\textsc{L2}). It follows that every marginal edge in $\mathcal{D}'$ is also a marginal edge in $\mathcal{D}$,}
\end{proof}

Given a tight bramble $\mathcal{B}$ of a graph $G$ of order at most $k$, a loose tree-decomposition $\mathcal{D}=(T,\chi)$ of $G$ is \emph{$\mathcal{B}$-admissible} if for every node $t\in V(T)$ such that $|\chi(t)|>k$, the bag $\chi(t)$ fails to cover $\mathcal{B}$. 
%Observe that the existence of a $\mathcal{B}$-admissible $(T,\chi)$ implies $\la(G) \leq k+1$. Indeed, as for every node $t\in T$, $\chi(t)$ covers $\mathcal{B}$ implying that $|\chi(t)| \leq k+1$. 

\begin{lemma} \label{th_bramble}
Let $G$ be a graph and $k\geq 1$ be an integer. If the maximum order of a tight bramble of $G$ is $k$, then $\ctp(G)\leq k$.
%If $\ctp(G)\geq k$, then $G$ has a tight bramble of order $k$.
\end{lemma}
\begin{proof}
%\xtof{For the backward direction, adapt Diestel's proof.}
We prove that if $G$ {does not contain a} tight bramble of order $k+1$ or larger, then for every tight bramble $\mathcal{B}$ of $G$, there exists a $\mathcal{B}$-admissible loose tree-decomposition of $G$. This statement implies $\ctp(G)\leq k$. Indeed, $\mathcal{B}=\emptyset$ is a tight bramble. As every subset of vertices covers the empty bramble, a $\emptyset$-admissible loose tree-decomposition cannot contain a bag of size larger than $k$. It follows by \autoref{th_ctp_width} that $\ctp(G)\leq k$.

%We assume that $G$ {does not contain a} tight bramble of order $k+1$ or larger. We show below that for every tight bramble $\mathcal{B}$ of $G$, there exists a $\mathcal{B}$-admissible loose tree-decomposition of $G$. We first observe that $\mathcal{B}=\emptyset$ is a tight bramble and that every bag of a $\emptyset$-admissible loose tree-decomposition has size at most $k$. Then, $\ctp(G)\leq k$ follows from  \autoref{th_ctp_width}. 

Let $\mathcal{B}$ be a non-empty tight bramble of $G$ and let $X$ be a cover of $\mathcal{B}$ of minimum size.  We proceed by induction on the size of $\mathcal{B}$. 
%\xtof{The proof of the base case of the induction is not correct.} If $\mathcal{B}=2^{V(G)}$, then $V(G)$ is a cover of minimum order and the loose tree-decomposition $\mathcal{D}$ with a unique bag containing $V(G)$ is $\mathcal{B}$-admissible. As by assumption the size of a cover is at most $k$, we obtain that $|V(G)| \leq k$ and then $\mathcal{D}$ certifies that $\ctp(G) \leq k$.

%Let $\mathcal{B}$ be a tight bramble of $G$ of maximum cardinality and let $X$ be a cover of $\mathcal{B}$ of minimum size.
%SED change 
{For the induction base}, suppose that the size of $\mathcal{B}$ is maximum.
Let $\mathsf{cc}(G-X)$ be the set of connected components of $G-X$. We build a $\mathcal{B}$-admissible loose tree-decomposition $\mathcal{D}=(T,\chi)$ of $G$ as follows. The tree $T$ is a star with $|\mathsf{cc}(G-X)|$ leaves. More precisely, we have $V(T)=\{t_X\}\cup\big\{t_C\mid C \in \mathsf{cc}(G-X)\big\}$ and $E(T)=\big\{t_Xt_C\mid C\in\mathsf{cc}(X)\big\}$. We set $\chi(t_X)=X$. It remains to define the bag $\chi(t_C)$ 
for every connected component $C \in \mathsf{cc}(G-X)$. Observe that since $C\cap X=\emptyset$ and since $X$ is a cover of $\mathcal{B}$, $C\notin\mathcal{B}$. Moreover, as $\mathcal{B}$ has maximum size, $\mathcal{B}\cup\big\{\{C\}\big\}$ is not a tight bramble of $G$. This implies 
that $\mathcal{B}$ contains a set $B$ such that $B$ and $C$ do not touch. It follows that $B\cap C=\emptyset$ and $G$ contains at most one edge $uv$ such that $u\in B\setminus C$ and $v\in C\setminus B$. If such an edge $uv$ exists, then we set $\chi(t_C)=(C\cup N_G(C))\setminus \{u\}$. Otherwise, we set $\chi(t_C)=C\cup N_G(C)$. In both cases, we have $\chi(t_C)\cap B=\emptyset$ and henceforth $\chi(t_C)$ does not cover $\mathcal{B}$. 
So if $\mathcal{D}=(T,\chi)$ is a loose tree-decomposition of $G$, then it is a $\mathcal{B}$-admissible one. Let us prove that the three conditions of \autoref{def_ltd} are satisfied. Condition \A{} follows from the fact that if a vertex $x$ of $G$ belongs to several bags, then it belongs to $\chi(t_X)$. Condition \B{} holds as $\chi(t_X)$ is a separator and for every connected component $C\in \mathsf{cc}(G-X)$, $\chi(t_C)\subseteq X\cup C$. Finally, condition \C{} is also satisfied because, by construction, if there exists a marginal edge between node $t_X$ and node $t_C$, then it is unique (as otherwise the sets $C$ and $B$ would touch).

%We now consider the case where $\mathcal{B}$ is not a tight bramble of maximum cardinality. 
{We now consider the case where the size of $\mathcal{B}$ is not maximum.}
By induction hypothesis, suppose that the statement holds for every tight bramble $\mathcal{B}'$ containing more elements than $\mathcal{B}$. 
We also assume that for every such larger tight bramble $\mathcal{B}'$, none of the $\mathcal{B}'$-admissible loose tree-decompositions is also $\mathcal{B}$-admissible, as otherwise we are done.
%Let $X \subseteq V(G)$ be a cover of $\mathcal{B}$ of minimum size, that is $|X|\leq k$. Observe that we can assume that $X \not = V(G)$ as otherwise {$G$ has a $\mathcal{B}$-admissible loose tree-decomposition containing a unique bag $V(G)$.} 
Observe that we can assume that $V(G)$ is not a minimum cover of $\mathcal{B}$, as otherwise $|V(G)|\leq k$ and thereby the loose tree-decomposition containing a unique bag $V(G)$ would be $\mathcal{B}$-admissible. By \autoref{th_ctp_width}, we then have $\ctp(G)\leq k$. So we assume that $X\neq V(G)$. For a set $S\subseteq V(G)$, we say that a loose tree-decomposition $\mathcal{D}=(T,\chi)$ of $G$ is \emph{$S$-rooted} if $T$ contains a node $t$ such that $S\subseteq \chi(t)$.

\begin{claim} \label{cl_amalgamate}
For every connected component $C$ of $G-X$, there exists a $X$-rooted loose tree-decomposition $\mathcal{D}_C=(T_C,\chi_C)$ of $H=G[C\cup X]$ such that for every node $t\in V(T_C)$ such that $|\chi_C(t)|\ge k+1$, $\chi(t)$ does not cover $\mathcal{B}$.
\end{claim}
%If this holds, then these $X$-rooted loose tree-decompositions can be amalgamated into a $\mathcal{B}$-admissible loose tree-decomposition of $G$ by identifying their root bags.

Suppose \autoref{cl_amalgamate} holds. Since for every $C\in cc(G-X)$, $\mathcal{D}_C=(T_C,\chi_C)$ contains a node $t_C$ such that $X\subseteq \chi_C(t_C)$, these loose tree-decompositions can be amalgamated into a loose tree-decomposition $\mathcal{D}=(T,\chi)$ of $G$ as follows: $T$ contains a node $t_X$ such that $\chi(t_X)=X$ and for each $C\in cc(G-X)$, $t_X$ is made adjacent to the node $t_C$ of $T_C$.

\medskip
\noindent
\textit{Proof of \autoref{cl_amalgamate}:}
{Let} $C$ be a connected component of $G-X$. We denote $H=G[X\cup C]$ and consider $\mathcal{B}'=\mathcal{B}\cup \{C\}$. We consider two cases:
\begin{itemize}
\item The set $\mathcal{B}'$ is not a tight bramble of $G$. Then $\mathcal{B}$ contains a set $B$ such that $B$ and $C$ are not touching, that is $B\cap C=\emptyset$ and $E(G)$ contains at most one edge $e$ that is incident to a vertex of $B$ and to a vertex of $C$. We proceed as in the induction base case above. Consider $\mathcal{D}_C=(T_C,\chi_C)$  where $T_C$ is the tree defined on two nodes $t_1$ and $t_2$. We set $\chi_C(t_1)=X$. If there is an edge $uv$ such that $u\in B\setminus C$ and $v\in C\setminus B$, then $\chi_c(t_2)=(C\cup N_G(C))\setminus \{u\}$, otherwise $\chi_c(t_2)=C\cup N_G(C)$. First observe that $\mathcal{D}_C$ is a loose tree-decomposition of $H$ satisfying the statement of the claim.
Indeed, by construction, as $B$ and $C$ are not touching. It follows that on the one hand, there is at most one marginal edge between $\chi(t_1)$ and $\chi(t_2)$, and so $\mathcal{D}_C=(T_C,\chi_C)$ is a loose tree-decomposition of $H$. On the other hand, this implies that $B\cap C=\emptyset$ and so $\chi(t_C)$ does not cover $\mathcal{B}$.
	
\item The set $\mathcal{B}'$ is a tight bramble of $G$. Since $X$ covers $\mathcal{B}$ and $C \cap X = \emptyset$, $C \not \in \mathcal{B}$. Therefore, we have $|\mathcal{B}'| > |\mathcal{B}|$. Consequently, by the induction hypothesis, we can assume the existence of a $\mathcal{B}'$-admissible loose tree-decomposition $\mathcal{D}=(T,\chi)$ of $G$.
%If this decomposition is also a $\mathcal{B}$-admissible one, we are done. \xtof{Revise the claim statement, to cover this conclusion. To be discussed with Dimitrios.} 
By the induction assumption, $\mathcal{D}$ is not $\mathcal{B}$-admissible. This implies that there exists a node $s\in V(T)$ such that $|\chi(s)|\geq k+1$ and $\chi(s)$ covers $\mathcal{B}$. 
%So assume that $\mathcal{D}$ is not $\mathcal{B}$-admissible, which means that there exists a node $s\in V(T)$ such that $|\chi(s)|\geq k+1$ and $\chi(s)$ covers $\mathcal{B}$. 
We build from $\mathcal{D}$ a $X$-rooted loose tree-decomposition $\mathcal{D}_C=(T_C,\chi_C)$ of $H=G[X\cup C]$ 
%\red{of width at most $\width(\mathcal{D},G)$} 
as follows (see \autoref{fig_amalgamate}):  for every vertex $x\in X$, let  $t_x\in V(T)$  be the node of $T_x$ that is the closest to $s$ in $T$. For every node $t\in V(T)$, we set:
$$\chi_C(t)=\big(\chi(t)\cap V(H)\big)\cup\big\{x\in X\mid t \mbox{ belongs to the unique $(t_x,s)$-path in $T$} \big\}.$$

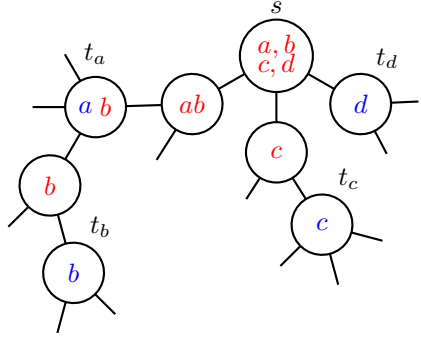
\begin{figure}[ht]
\begin{center}
\begin{tikzpicture}[thick,scale=1]
\tikzstyle{sommet}=[circle, draw, fill=black, inner sep=0pt, minimum width=3pt]

\begin{scope}[xshift=5cm,scale=0.8]

\node[] (0) at (0:0) {};
\node[] at (90:0.8) {$s$};

\node[] (1) at (-150:1.6) {};
\node[] (2) at (-90:1.6) {};
\node[] (3) at (-30:1.6) {};
\foreach \i in {1,2,3}{
\draw (0) -- (\i) ;
}
\draw[fill=white] (0) circle (0.6) ;
\node[color=red] at (0.north) {$a,b$};
\node[color=red] at (0.south) {$c,d$};

%---------------------------------------------------
\begin{scope}[shift=(-150:1.7),rotate=-60]
\node[] (20) at (-120:1.5) {};
\node[] (21) at (-60:1.2) {};
\draw (1) -- (20) ;
\draw (1) -- (21) ;
\draw[fill=white] (1) circle (0.5) ;
\node[color=red] at (1) {$ab$};

%%--------------------------------------------------
\begin{scope}[shift=(-120:1.5),rotate=-30]
\node[] (30) at (-150:1.2) {};
\node[] (31) at (-90:1.2) {};
\node[] (32) at (-30:1.5) {};
\draw (20) -- (30);
\draw (20) -- (31);
\draw (20) -- (32);
\draw[fill=white] (20) circle (0.5) ;
\node[color=blue] at (20.west) {$a$};
\node[color=red] at (20.east) {$b$};
\node[shift=(90:0.7)] at (20) {$t_a$};

%%%-------------------------------------------------
\begin{scope}[shift=(-30:1.5),rotate=75]
\node[] (40) at (-120:1.2) {};
\node[] (41) at (-60:1.5) {};
\draw (32) -- (40);
\draw (32) -- (41);
\draw[fill=white] (32) circle (0.5) ;
\node[color=red] at (32) {$b$};

%%%%-----------------------------------------------
\begin{scope}[shift=(-60:1.5),rotate=30]
\node[] (50) at (-120:1.2) {};
\node[] (51) at (-60:1.2) {};
\draw (41) -- (50);
\draw (41) -- (51);
\draw[fill=white] (41) circle (0.5) ;
\node[color=blue] at (41) {$b$};
\node[shift=(60:0.7)] at (41) {$t_b$};
\end{scope}
%%%%-----------------------------------------------

\end{scope}
%%%-------------------------------------------------

\end{scope}
%%--------------------------------------------------
\end{scope}
%---------------------------------------------------

%---------------------------------------------------
\begin{scope}[shift=(-90:1.5),rotate=0]
\node[] (22) at (-120:1.2) {};
\node[] (23) at (-60:1.5) {};
\draw (2) -- (22);
\draw (2) -- (23);
\draw[fill=white] (2) circle (0.5) ;
\node[color=red] at (2) {$c$};

%%--------------------------------------------------
\begin{scope}[shift=(-60:1.5),rotate=15]
\node[] (33) at (-150:1.2) {};
\node[] (34) at (-90:1.2) {};
\node[] (35) at (-30:1.2) {};
\draw (23) -- (33);
\draw (23) -- (34);
\draw (23) -- (35);
\draw[fill=white] (23) circle (0.5) ;
\node[color=blue] at (23) {$c$};
\node[shift=(60:0.7)] at (23) {$t_c$};
\end{scope}
%%--------------------------------------------------
\end{scope}
%---------------------------------------------------

%---------------------------------------------------
\begin{scope}[shift=(-30:1.5),rotate=60]
\node[] (24) at (-120:1.2) {};
\node[] (25) at (-60:1.2) {};
\draw (3) -- (24);
\draw (3) -- (25);
\draw[fill=white] (3) circle (0.5) ;
\node[color=blue] at (3) {$d$};
\node[shift=(60:0.7)] at (3) {$t_d$};
\end{scope}
%---------------------------------------------------

\end{scope}

\end{tikzpicture}
\end{center}
\caption{Suppose that $X=\{a,b,c,d\}$ and that node $s$ has size $|\chi(s)|>k$. The nodes $t_a$, $t_b$, $t_c$ and $t_d$ are respectively the nodes of the traces $T_a$, $T_b$, $T_c$ and $T_d$ that are the closest to $s$. To construct $\mathcal{D}_C$, vertices $a$, $b$, $c$ and $d$ are respectively added to the the bag on the paths between $s$ and $t_a$, $t_b$, $t_C$ and $t_d$. \label{fig_amalgamate}}
\end{figure}

Finally $T_C$ is obtained by restricting $T$ to the nodes $V(T_C)=\{t\in V(T)\mid \chi_C(t)\neq\emptyset\}$. Observe that by construction of $\chi_C$, for every vertex $x\in X\cup C$, the set ${T_{C,x}=}\{t\in V(T_C)\mid x\in\chi_C(t)\}$ contains the trace $T_x$ of $x$ in $\mathcal{D}$ {and the unique path in $T$ between the trace $T_x$ and the node $s$. It follows that $T_{C,x}$ is a connected subtree of $T$ and so \autoref{lem_td_extension} applies. Consequently $\mathcal{D}_C=(T_C,\chi_C)$ of $H=G[X\cup C]$ is a loose tree-decomposition of $H$.} 
%It follows that \autoref{lem_td_extension} applies, and consequently $\mathcal{D}_C=(T_C,\chi_C)$ of $H=G[X\cup C]$ is a loose tree-decomposition of $H$. 

We now prove that $\mathcal{D}_C$ satisfies the statement of the claim.
To that aim, we first prove that for every node $t\in V(T_C)$, $|\chi_C(t)|\leq |\chi(t)|$. As $\chi(s)$ and $X$ both cover the tight bramble $\mathcal{B}$, by \autoref{bramble_separation}, any separator between $X$ and $\chi(s)$ also covers $\mathcal{B}$. By the minimality of $|X|$, such a separator has size at least $|X|$. It follows by Menger Theorem that there exists a set $\mathcal{P}$ of $|X|$ vertex disjoint paths between $X$ and $\chi(s)$. For every $x\in X$, we  denote by $P_x$ the path of $\mathcal{P}$ containing $x$. As $(T,\chi)$ is $\mathcal{B}'$-admissible and $|\chi(s)| \geq k+1$, $\chi(s)$ fails to cover $\mathcal{B}'$. Since $\chi(s)$ covers $\mathcal{B}$ and since $\mathcal{B}' = \mathcal{B} \cup \{C\}$, it follows that $\chi(s)$ does not intersect $C$. It follows that every path $P_x\in\mathcal{P}$ belongs to $G-C$ (see \autoref{fig_disjoint_path}). 
%\xtof{Referee 3 asks for some details about why every path $P_x\in\mathcal{P}$ belongs to $G-C$. It seems to be obvious to me. I will add a figure} 
Let us now consider a node $t$ for which there exists a vertex $x\in\chi_C(t)\setminus\chi(t)$. Observe that $t$ lies in $T_C$ on the unique $(t_x,s)$-path. By \autoref{lem_node_sep}, $\chi(t)$ separates $\chi(t_x)$ and $\chi(s)$. It follows that $\chi(t)$ contains a vertex of $P_x$. As $P_x$ is a path of $G-C$, for every vertex $x\in\chi_C(t)\setminus\chi(t)$, we identify a vertex $y\in\chi(t)\setminus\chi_C(t)$. It follows that $|\chi_C(t)|\leq |\chi(t)|$.

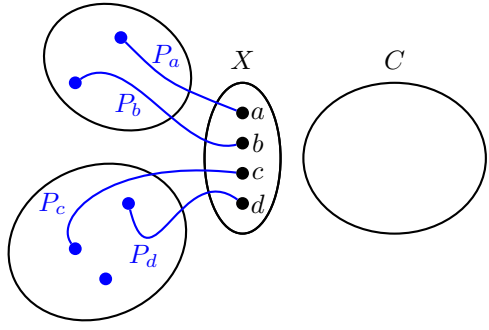
\begin{figure}[ht]
\begin{center}
\begin{tikzpicture}[thick,scale=1]
\tikzstyle{sommet}=[circle, draw, fill=black, inner sep=0pt, minimum width=4pt]
\tikzstyle{bsommet}=[circle, draw, fill=blue, inner sep=0pt, minimum width=4pt]

\draw (0,0) ellipse(0.5 and 1);
\draw (2,0) ellipse(1.2 and 1);
\draw[rotate=-25] (-2,0.4) ellipse(1 and 0.8);
\draw [rotate=25] (-2.2,-0.2) ellipse(1.2 and 1);
\draw (0,0) ellipse(0.5 and 1);
\draw node at (0,1.3) {$X$};
\draw node at (2,1.3) {$C$};

\draw[] node[sommet] (a) at (0,0.6){};
\draw[] node[sommet] (b) at (0,0.2){};
\draw[] node[sommet] (c) at (0,-0.2){};
\draw[] node[sommet] (d) at (0,-0.6){};
\draw node at (0.2,0.6) {$a$};
\draw node at (0.2,0.2) {$b$};
\draw node at (0.2,-0.2) {$c$};
\draw node at (0.2,-0.6) {$d$};

\draw[color=blue] node[bsommet] (S1) at (-2.2,1){};
\draw[color=blue] node[bsommet] (S2) at (-1.6,1.6){};

\draw[color=blue] node[bsommet] (S3) at (-1.8,-1.6){};
\draw[color=blue] node[bsommet] (S4) at (-2.2,-1.2){};
\draw[color=blue] node[bsommet] (S5) at (-1.5,-0.6){};

\draw[color=blue] (a) .. controls (-1,1) .. (S2);
\draw[color=blue] node at (-1,1.35) {$P_a$};
\draw[color=blue] (b) .. controls (-0.8,0) and (-1.5,1.5) .. (S1);
\draw[color=blue] node at (-1.5,0.7) {$P_b$};
\draw[color=blue] (c) .. controls (-1.8,-0) and (-2.5,-0.8) .. (S4);
\draw[color=blue] node at (-2.5,-0.6) {$P_c$};
\draw[color=blue] (d) .. controls (-0.8,0) and (-1.2,-1.8) .. (S5);
\draw[color=blue] node at (-1.3,-1.3) {$P_d$};
\end{tikzpicture}
\end{center}
\caption{The set $\chi(s)$ is disjoint from $C$ and $|X|\leq\chi(s)$. The paths $P_a$, $P_b$, $P_c$ and $P_d$ are pairwise vertex disjoint and do not intersect the component $C$. \label{fig_disjoint_path}}
\end{figure}

To conclude, let $t\in V(T_C)$ be a node such that $|\chi_C(t)|>k$. Since $|X| \leq k$ and since $\chi_C(t) \subseteq X\cup C$, we have that $\chi_C(t)\cap C\neq\emptyset$. 
As $(T,\chi)$ is $\mathcal{B}'$-admissible, $\chi(t)$ does not intersect some $B \in \mathcal{B}$. Let us show that $\chi_C(t)$ does not intersect $B$ either. Assume it does. Then there exists $x\in B$ such that $x \in \chi_C(t) \setminus \chi(t)$. Observe that by construction of $\chi_C(t)$, $x\in X$ and thereby $\chi(t_x)\cap B\neq\emptyset$. As $\chi(s)$ covers $\mathcal{B}$, we also have that $\chi(s)\cap B\neq\emptyset$. By \autoref{lem_node_sep}, $\chi(t)$ separates $\chi(s)$ and $\chi(t'_x)$. Since $B$ is a connected subset of vertices, $\chi(t)\cap B\neq\emptyset$, a contradiction. Finally, as by construction every vertex of $X$ belongs to $\chi_C(s)$, $\mathcal{D}_C$ is $X$-rooted. The claim follows \hfill $\diamond$
%\xtof{Referee 3 asks to explain why $\mathcal{D}_C$ is $X$-rooted. For this maybe we should add a bag containing X adjacent to $s$ as the root. What do you think ?}
\end{itemize}
\end{proof}

%---------------------------------------------------------------------------------------------------------------
\subsection{Escape strategies derived from tight brambles}

We now show how {given} a tight bramble of order $k$ in a graph $G$, a {fugitive can derive an escape} strategy $(e_1,\f_G)\in\mathcal{F}_G$ such that for every search strategy $\s_G\in \mathcal{S}_G$ of cost $k$, the play $\play(\s_G,e_1,\f_G)$ generated by the program $\big(\s_G,(e_1,\f_G)\big)$ is infinite. In other words, the fugitive cannot be captured by a set of $k$ searchers. This proves $(4\Rightarrow 3)$ from \autoref{th_min_max}.

\begin{lemma} \label{th_escape}
Let $G$ be a graph and {$k\geq 1$} be an integer. If $G$ has a tight bramble of order $k$, then $\avms(G)\geq k$.
\end{lemma}
\begin{proof}
{Let us first consider the case $k=1$. By \autoref{obs_bramble_tree}, trees are the connected graphs for which the maximum order of a tight bramble is $1$. Observe that, in a tree $T$, one searcher is enough to capture the fugitive. Suppose that the unique searcher is located at vertex $x$ and that the fugitive is located at edge $e$. Then the search strategy consists in sliding along the unique edge $xy$ incident to $x$ towards $e$: that is, $\s_T(\{x\},e)=y$ where $y$ is the unique neighbor of $x$ in the connected component of $T-x$ containing $e$. Clearly, eventually the fugitive will be located on an edge incident to a leaf of $T$ and will be captured.
}

{So let us consider $\mathcal{B}$ a tight bramble of $G$ of order $k\geq 2$.}
%Let $\mathcal{B}$ be a tight bramble of $G$ of order \cor{$k$. 
Let $\s_G\in \mathcal{S}_G\subseteq \big(2^{V(G)}\big)^{\big(2^{V(G)}\times E(G)\big)}$ be an arbitrary search strategy such that $\cost(\s_G)<k$. 
Let $B_0$ be an arbitrary tight bramble element and let $e_1$ be an edge of $E(B_0)$.
We build a fugitive strategy $(e_1,\f_G)\in\mathcal{F}_G$ such that the play $\play(\s_G,e_1,\f_G)=\langle S_0,e_1,S_1, \dots, S_{i-1},e_i,S_i,\dots \rangle$ generated by the program $\big(\s_G,(e_1,\f_G)\big)$ 
verifies that for every $i\geq 1$, $\rspace_G(S_{i-1},e_{i-1},S_i)\neq\emptyset$. 
Such a fugitive strategy $(e_1,\f_G)$ certifies that the  fugitive is never captured and thereby $\s_G$ is not winning.
To that aim, we show that for every $i\geq 1$, there exists $B_i\in\mathcal{B}$ such that $S_i\cap B_i=\emptyset$ and {$B_i\subseteq A_G(S_{i-1},e_i,S_i)$}. 
This allows us to select $e_{i+1}=\f_G(S_{i-1},e_i,S_i)\in B_i$.

%\xtof{In the induction below, we don't consider the case of $k=1$, that is, the size of the bramble cover is $1$. In that case, we can prove that $G$ is acyclic and can thereby be search with one searcher by iteratively sliding along an edge towards the visible robber. So the statement holds. Does similar particular case has to be consider in other parts of the paper ?}

We proceed by induction on $i$. Clearly, {since $k\geq 2$,} the property is satisfied for $i=1$. Indeed, as $S_0=\emptyset$, we have that $|S_1|=1$. Thereby $S_1$ is not a cover of $\mathcal{B}$. So there exists $B_1\in\mathcal{B}$ which is disjoint from $S_1$. 
%We can assume that $e_1\in E(B_0)$ for some subset of vertices $B_0\in\mathcal{B}$.  
Suppose that $B_0\neq B_1$ (otherwise, we are done). As $B_0$ and $B_1$ are touching connected subsets of vertices of $G$ and as $|S_1|=1$, we have that {$B_1\in A_G(S_0,e_1,S_1)$.} It follows that we can set $\f_G(S_0,e_1,S_1)=e_2$ with $e_2$ being any edge of $E(B_1)$.
Suppose that the property holds for every $j$ such that $1\leq j<i$. %By assumption, we know that $|S_i|<k$ implying that $S_i$ is not a cover of $\mathcal{B}$. Let $B_i\in\mathcal{B}$ be a subset of vertices such that $B_i\cap S_i=\emptyset$. 
We distinguish three cases:
\begin{itemize}
\item $S_i=S_{i-1}\setminus\{u\}$ (that is, a searcher is removed from vertex $u$): It follows that $S_i\cap B_{i-1}=\emptyset$ and thereby we can set $B_i=B_{i-1}$ and $\f_G(S_{i-1},e_i,S_i)=e_{i+1}$ with $e_{i+1}$ being any edge of $E(B_i)$. Observe that the fugitive may not move at such a step.

\item $S_i=S_{i-1}\cup\{v\}$ (that is, a new searcher is placed on vertex $v$): Observe that if $v\notin B_{i-1}$, then $B_{i-1}\cap S_i=\emptyset$. As $e_i\in E(B_{i-1})$ and $B_{i-1}$ is connected, we also have that $B_{i-1}\subseteq A_G(S_{i-1},e_i,S_i)$. Thereby we can set $B_i=B_{i-1}$ and $\f_G(S_{i-1},e_i,S_i)=e_{i+1}$ with $e_{i+1}$ being any edge of $E(B_i)$. So assume that $v\in B_{i-1}$. As $|S_i|<k$, $S_i$ is not a cover of $\mathcal{B}$ and there exists $B_i\in\mathcal{B}$ such that $S_i\cap B_i=\emptyset$. Now observe that $S_i\cap B_{i-1}=\{v\}$. It follows that as $B_{i-1}$ and $B_i$ are touching, every edge in $E(B_i)$ is accessible through a $(S_{i-1},S_i)$-avoiding pathway from $e_i$, that is {$B_i\subseteq A_G(S_{i-1},e_i,S_i).$} Thereby we can set $\f_G(S_{i-1},e_i,S_i)=e_{i+1}$ with $e_{i+1}$ being any edge of $E(B_i)$.

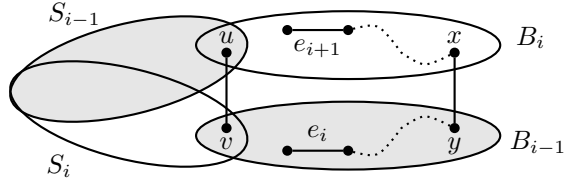
\begin{figure}[ht]
\begin{center}
\begin{tikzpicture}[thick,scale=1]
\tikzstyle{sommet}=[circle, draw, fill=black, inner sep=0pt, minimum width=3pt]

\draw[fill=gray!20] (1.6,-1.1) circle (2 and 0.45);
\draw[rotate=14,fill=gray!20] (-1.3,0.1) circle (1.6 and 0.6);

\draw[] (1.6,0.1) circle (2 and 0.45);
\draw[rotate=-14] (-1.05,-1.1) circle (1.6 and 0.6);

\draw node[sommet] (u) at (0,0){};
\draw node[sommet] (v) at (0,-1){};
\draw node[sommet] (x) at (3,0){};
\draw node[sommet] (y) at (3,-1){};

\node[above] at (0,0) {$u$};
\node[below] at (0,-1) {$v$};
\node[above] at (3,0) {$x$};
\node[below] at (3,-1) {$y$};

\draw (u) -- (v);
\draw (x) -- (y);

\draw node[sommet] (a1) at (0.8,0.3){};
\draw node[sommet] (a2) at (1.6,0.3){};
\draw (a1) -- (a2) ;
\node[below] at (1.2,0.3) {$e_{i+1}$};

\draw node[sommet] (b1) at (0.8,-1.3){};
\draw node[sommet] (b2) at (1.6,-1.3){};
\draw (b1) -- (b2) ;
\node[above] at (1.2,-1.3) {$e_{i}$};

\draw[dotted]  (a2) .. controls (2,0.4) .. (2.3,0) .. controls (2.6,-0.2) .. (x);
\draw[dotted]  (b2) .. controls (2,-1.4) .. (2.3,-1) .. controls (2.6,-0.8) .. (y);

\node[] at (-2,0.5) {$S_{i-1}$};
\node[] at (4,0.2) {$B_{i}$};
\node[] at (-2.2,-1.5) {$S_{i}$};
\node[] at (4.1,-1.2) {$B_{i-1}$};
\end{tikzpicture}
\end{center}
\caption{\label{fig_slide} Two consecutive sets $S_{i-1}$ and $S_i$ of vertices occupied by the searchers from the play $\play$ such that $S_i\ominus S_{i-1}=\{u,v\}\in E(G)$ is an edge of multiplicity one. The tight ramble element $B_{i-1}\in\mathcal{B}$ that is disjoint from $S_{i-1}$ 
contains the edge $e_i$. Likewise the tight bramble element $B_{i}\in\mathcal{B}$ that is disjoint from $S_{i}$ contains the edge $e_{i+1}$. We observe that if $B_{i-1}$ and $B_i$ are disjoint, as they are touching, then there exists a pathway from $e_i$ to $e_{i+1}$ going through the edge $xy$ that avoids the edge $uv$.}
\end{figure}

\item $S_i\ominus S_{i-1}=\{u,v\}$ is an edge of $E(G)$ (that is, a searcher slides on the edge $uv$ towards 
$v$): As in the previous case, observe that if $v\notin B_{i-1}$, then $B_{i-1}\cap S_i=\emptyset$. As $e_i\in 
E(B_{i-1})$ and $B_{i-1}$ is connected, we also have that $B_{i-1}\subseteq A_G(S_{i-1},e_i,S_i)$. Thereby we 
can set $B_i=B_{i-1}$ and $\f_G(S_{i-1},e_i,S_i)=e_{i+1}$ with $e_{i+1}$ being any edge of $E(B_i)$. So
assume that $v\in B_{i-1}$. As $|S_i|<k$, $S_i$ is not a cover of $\mathcal{B}$ implying the existence of 
$B_i\in\mathcal{B}$ such that $S_i\cap B_i=\emptyset$ (and so $v\notin B_i$). If $B_i\cap B_{i-1}\neq\emptyset$, then by the connectivity of the tight bramble elements, clearly every edge in $E(B_i)$ is 
accessible through a $(S_{i-1},S_i)$-avoiding pathway from $e_i$. Suppose that $B_i\cap B_{i-1}=\emptyset$ 
(see \autoref{fig_slide}). As $B_i$ and $B_{i-1}$ are touching, there exists an edge $xy$ distinct from $\{u,v\}$ 
such that $x\in B_i$ and $y\in B_{i-1}$. Again, the connectivity of the tight bramble elements implies that every 
edge in $E(B_i)$ is accessible through a $(S_{i-1},S_i)$-avoiding pathway from $e_i$ containing the edge $xy$. 
So in both cases, we have {$B_i\subseteq A_G(S_{i-1},e_i,S_i).$} 
Thereby we can set $\f_G(S_{i-1},e_i,S_i)=e_{i+1}$ with $e_{i+1}$ being any edge of $E(B_i)$. 
\vspace{-6mm}
%\xtof{Observe that the use one the multiplicity one of the edge $\{u,v\}$ is implicitly used here in the definition of $S_{i-1},S_i)$-avoiding path. Shall we make it more explicit ?}
\end{itemize}	
\end{proof}

%---------------------------------------------------------------------------------------------------------------
\subsection{Monotone search strategies derived from loose tree-decompositions}

We show that a loose tree-decomposition of width $k$ for a graph $G$ can be used to design a winning search strategy $\s_G\in \mathcal{S}_G$ of cost $k$ {that is monotone}. This establishes $(2\Rightarrow 5)$ from~\autoref{th_min_max}.

\begin{lemma}\label{th_avms_la}
Let $G$ be a graph and $k\geq 0$ be an integer. If $\ctp(G)\leq k$, then $\mavms(G)\leq k$.
\end{lemma}
\begin{proof}
Suppose that $\ctp(G)\leq k$. By \autoref{th_ctp_width}, $G$ admits a full loose tree-decomposition $\mathcal{D}=(T,\chi)$ such that $\width(\mathcal{D})\leq k$. We build from $\mathcal{D}$ a search strategy 
$\s_G\in\mathcal{S}_G\subseteq \big(2^{V(G)}\big)^{\big(2^{V(G)}\times E(G)\big)}$ and prove that  it is a winning strategy of cost at most $k$. Consider an arbitrary  fugitive strategy $(e_1,\f_G)\in\mathcal{F}_G$ and let $\play(\s_G,e_1,\f_G)=\langle S_0, e_1, S_1,\dots, S_{i-1}, e_i, S_i,\dots\rangle$ 
be the play generated by the program
%defined by 
$\big(\s_G,(e_1,\f_G)\big)$. For $i\geq 1$, the searchers' position $S_i=\s_G(S_{i-1},e_i)$ is defined as follows. {At every step $i$, we guarantee the existence of a node $t\in V(T)$ such that $S_i\subseteq \chi(t)$.} Pick a vertex $x\in V(G)$ and consider a node $t\in V(T)$ such that $x\in\chi(t)$. We set $S_1=\{x\}$ and for $1<i\leq \ell=|\chi(t)|$, $S_i=S_{i-1}\cup\{y\}$ for some $y\in \chi(t)\setminus S_{i-1}$. That is, up to step $\ell$, we iteratively add a searcher on every vertex of $\chi(t)$.
Suppose that $i> \ell$. Let $t=V(T)$ be the node such that $S_{i-1}\subseteq \chi(t)$ and let $f=\{t_1,t_2\}$ be the tree-edge of $T$ such that {$e_i=\f_G(S_{i-2},e_{i-1},S_{i-1})\in E(G[\chi(t_1)\cup\chi(t_2)]).$} Suppose that $t_1$ is closer to $t$ in $T$ than $t_2$ is. Let $t'$ be the neighbour of $t$ in $T$ such that $t'$ and $t_2$ are in the same connected component of $T-tt'$. Observe that we may have $t=t_1$ and thereby $t'=t_2$. As $\mathcal{D}$ is full, there are $u\in \chi(t)$ and $v\in\chi(t')$ such that $\{u,v\}=\chi(t)\ominus\chi(t')$. There are two cases to consider:
\begin{enumerate}
\item if $uv\in E(G)$,
%(and by condition \C{} of \autoref{def_ltd}, there is no \blue{multiple edge} between $u$ and $v$), 
then  $S_i=\s_G(S_{i-1},e_i)=\chi(t')$ is obtained by sliding along the edge $uv$ (towards $v$);
\item if $uv\notin E(G)$, then $S_i=\s_G(S_{i-1}, e_i)=\chi(t)\setminus\{u\}$ is obtained by removing a searcher from vertex $u$, and $S_{i+1}=\s_G(S_{i}, e_{i+1})=\chi(t')$, with $e_{i+1}=\f_G(S_{i-1},e_i,S_i)$, is obtained by adding a searcher on vertex $v$.
\end{enumerate}
Let us prove that $\s_G$ is a winning search strategy. To that aim, we show that there is some step $i\geq 1$ such that {$\rspace_G(S_{i-1},e_{i},S_{i})=\emptyset.$} % (and thereby $e_i\in C_i$ and $A_G(S_{i-1},S_i,e_{i-1})=\emptyset$). 
First observe, that if at step $i$ a searcher is added, that is $S_i=S_{i-1}\cup\{u\}$ for some $u\in V(G)$, then {$\rspace_G(S_{i-1},e_i,S_{i})\subsetneq \rspace_G(S_{i-2},e_{i-1},S_{i-1}).$} Now 
suppose that at step $i$, a searcher slides on the edge $uv$  where $\{u,v\}=\chi(t)\ominus\chi(t')$. By \autoref{lem_node_sep}, $\chi(t)$ is a separator of $G$. Observe that, by construction of $\s_G$, either
$e_i=uv$ or $e_i$ and $v$ belongs to the same component of $G-\chi(t)$. It follows that {$\rspace_G(S_{i-1},e_i,S_{i})\subsetneq \rspace_G(S_{i-2},e_{i-1},S_{i-1}).$} 
%\xtof{more details needed ?}
Finally, suppose that at step $i$, a searcher is removed from a vertex $u\in S_{i-1}$. By construction of $\s_G$, this is only the case where the set $\{u,v\}=\chi(t)\ominus\chi(t')$ does not correspond to an edge of $E(G)$. By \autoref{lem_tree_edge_sep}, $\chi(t)\cap\chi(t')$ is a separator of $G$. 
Observe in that case that {$\rspace_G(S_{i-1},e_{i},S_{i})= \rspace_G(S_{i-2},e_{i-1},S_{i-1}).$} But as, by construction, $\s_G$ 
does not contain two consecutive steps removing a searcher, we can conclude that there exists some step $i$ such that {$\rspace_G(S_{i-1},e_{i},S_{i})=\emptyset.$} 

Moreover observe that, from the discussion above, for every $(e_1,\f_G)\in\mathcal{F}_G$, the play generated by the program $\big(\s_G,(e_1,\f_G)\big)$ is monotone. This, in turns, implies that the strategy $\s_G$ is monotone. Finally we observe that by the  construction of $\s_G$, we have $\cost(\s_G)=\ctp(G)\leq k$, proving the result.
\end{proof}

\section{Discussion}

In this paper, we defined the concepts of loose tree-decomposition and tight bramble. We showed that the existence of a tight bramble of large order in a graph is an obstacle to a loose tree-decomposition of small width. Moreover, we proved the equivalence of the corresponding parameters to the Cartesian product number, introduced as the \emph{largeur d'arboresence} by Colin de Verdi\`ere~\cite{ColinDeVerdiere98multi} (see also \cite{Harvey14ontree}) and the mixed search number against an agile and visible fugitive (\autoref{th_min_max}).

One may consider path-like counterparts  of all these concepts as follows.  
We can restrict the definition of loose tree-decomposition to path instead of tree, then we obtain {\em loose path-decompositions}. Similarly, if in the definition of  Cartesian tree product number, we again consider trees instead of paths, this yields the definiton of the \emph{Cartesian path product number}. 
A ``path analogue'' of \autoref{th_min_max} may easily be produced by considering the 
{\em mixed search number against an agile and invisible fugitive} introduced by Bienstock and Seymour in~\cite{BienstockS91Monot}, while the min-max analogue for these path-like parameters
can be derived by the framework of \cite{FominT03Onthe}. 
Finally, we wish to mention that
a different  variant for path-decomposition that is equivalent the above parameters was introduced 
by Takahashi and  Ueno and  Kajitani in~\cite{TakahashiUK95Mixed,TakahashiUK94minim} under the name {\sl proper path decomposition}.

%SED change
Finally, our results may support the belief that it might be  possible to derive the min-max duality of \autoref{th_min_max} using general parameter duality frameworks such as those in   \cite{Pardo13purs,LyaudetMT10,AminiMNT09subm,MazoitN08mono}.

\paragraph{Acknowledgments.} We would like to thank the anonymous referees for their careful reading and accurate remarks that helped us to greatly improve the presentation of the manuscript.

%\input{SECTIONS/tight_brambles}

%\bibliography{largeur_arborescente.bib}
%\bibliographystyle{plainurl}
%
%
%\end{document}

%---------------------------------------------------------------------------------------------------------------
%---------------------------------------------------------------------------------------------------------------
%\nocite{*}
\bibliographystyle{plainurl}
%\bibliography{largeur_arborescente.bib}

\begin{thebibliography}{10}

\bibitem{AminiMNT09subm}
Omid Amini, Fr{\'{e}}d{\'{e}}ric Mazoit, Nicolas Nisse, and St{\'{e}}phan
  Thomass{\'{e}}.
\newblock Submodular partition functions.
\newblock {\em Discret. Math.}, 309(20):6000--6008, 2009.
\newblock \href {https://doi.org/10.1016/j.disc.2009.04.033}
  {\path{doi:10.1016/j.disc.2009.04.033}}.

\bibitem{BellenbaumD02Twosh}
Patrick Bellenbaum and Reinhard Diestel.
\newblock Two short proofs concerning tree-decompositions.
\newblock {\em Combinatorics, Probability and Computing}, 11(6):541--547, 2002.
\newblock \href {https://doi.org/10.1017/S0963548302005369}
  {\path{doi:10.1017/S0963548302005369}}.

\bibitem{BienstockRST91Quick}
Daniel Bienstock, Niel Robserston, Paul~D. Seymour, and Robin Thomas.
\newblock Quickly excluding a forest.
\newblock {\em Journal of Combinatorial Theory, Series B}, 52:274--283, 1991.
\newblock \href {https://doi.org/10.1016/0095-8956(91)90068-U}
  {\path{doi:10.1016/0095-8956(91)90068-U}}.

\bibitem{BienstockS91Monot}
Daniel Bienstock and Paul~D. Seymour.
\newblock Monotonicity in graph searching.
\newblock {\em Journal of Algorithms}, 12(2):239--245, 1991.
\newblock \href {https://doi.org/10.1016/0196-6774(91)90003-H}
  {\path{doi:10.1016/0196-6774(91)90003-H}}.

\bibitem{DendrisKT97Fugit}
Nick~D. Dendris, Lefteris~M. Kirousis, and Dimitrios~M. Thilikos.
\newblock Fugitive-search games on graphs and related parameters.
\newblock {\em Theoretical Computer Science}, 172:233--254, 1997.
\newblock \href {https://doi.org/10.1016/S0304-3975(96)00177-6}
  {\path{doi:10.1016/S0304-3975(96)00177-6}}.

\bibitem{EllisST94Theve}
Jonathan~A. Ellis, Ivan~Hal Sudborough, and Jonathan~S. Turner.
\newblock The vertex separation and search number of a graph.
\newblock {\em Information and Computation}, 113(1):50--79, 1994.
\newblock \href {https://doi.org/10.1006/inco.1994.1064}
  {\path{doi:10.1006/inco.1994.1064}}.

\bibitem{FominT03Onthe}
Fedor~V. Fomin and Dimitrios~M. Thilikos.
\newblock On the monotonicity of games generated by symmetric submodular
  functions.
\newblock {\em Discret. Appl. Math.}, 131(2):323--335, 2003.
\newblock \href {https://doi.org/10.1016/S0166-218X(02)00459-6}
  {\path{doi:10.1016/S0166-218X(02)00459-6}}.

\bibitem{FominT08Anann}
Fedor~V. Fomin and Dimitrios~M. Thilikos.
\newblock An annotated bibliography on guaranteed graph searching.
\newblock {\em Theoretical Computer Science}, 399:236--245, 2008.
\newblock \href {https://doi.org/10.1016/j.tcs.2008.02.040}
  {\path{doi:10.1016/j.tcs.2008.02.040}}.

\bibitem{Halin76S-func}
Rudolph Halin.
\newblock {$S$}-functions for graphs.
\newblock {\em Journal of Geometry}, 8:171--186, 1976.
\newblock \href {https://doi.org/10.1007/BF01917434}
  {\path{doi:10.1007/BF01917434}}.

\bibitem{Harvey14ontree}
Daniel~J. Harvey.
\newblock {\em On treewidth and graph minors}.
\newblock PhD thesis, University of Melbourne, 2014.

\bibitem{Kinnersley92Theve}
Nancy~G. Kinnersley.
\newblock The vertex separation number of a graph equals its path-width.
\newblock {\em Information Processing Letters}, 42(6):345--350, 1992.
\newblock \href {https://doi.org/10.1016/0020-0190(92)90234-M}
  {\path{doi:10.1016/0020-0190(92)90234-M}}.

\bibitem{KirousisP86Searc}
Lefteris~M. Kirousis and Christos~H. Papadimitriou.
\newblock Searching and pebbling.
\newblock {\em Theoretical Computer Science}, 47(2):205--218, 1986.
\newblock \href {https://doi.org/10.1016/0304-3975(86)90146-5}
  {\path{doi:10.1016/0304-3975(86)90146-5}}.

\bibitem{LyaudetMT10}
Laurent Lyaudet, Fr{\'{e}}d{\'{e}}ric Mazoit, and St{\'{e}}phan Thomass{\'{e}}.
\newblock Partitions versus sets: {A} case of duality.
\newblock {\em Eur. J. Comb.}, 31(3):681--687, 2010.
\newblock \href {https://doi.org/10.1016/j.ejc.2009.09.004}
  {\path{doi:10.1016/j.ejc.2009.09.004}}.

\bibitem{MazoitN08mono}
Fr{\'{e}}d{\'{e}}ric Mazoit and Nicolas Nisse.
\newblock Monotonicity of non-deterministic graph searching.
\newblock {\em Theor. Comput. Sci.}, 399(3):169--178, 2008.
\newblock \href {https://doi.org/10.1016/j.tcs.2008.02.036}
  {\path{doi:10.1016/j.tcs.2008.02.036}}.

\bibitem{Nisse19Netwo}
Nicolas Nisse.
\newblock Network decontamination.
\newblock In Paola Flocchini, Giuseppe Prencipe, and Nicola Santoro, editors,
  {\em Distributed Computing by Mobile Entities, Current Research in Moving and
  Computing}, volume 11340 of {\em Lecture Notes in Computer Science}, pages
  516--548. Springer, 2019.
\newblock \href {https://doi.org/10.1007/978-3-030-11072-7\_19}
  {\path{doi:10.1007/978-3-030-11072-7\_19}}.

\bibitem{Pardo13purs}
Ronan Pardo~Soares.
\newblock {\em {Pursuit-Evasion, Decompositions and Convexity on Graphs}}.
\newblock Theses, {Universit{\'e} Nice Sophia Antipolis}, November 2013.
\newblock URL: \url{https://tel.archives-ouvertes.fr/tel-00908227}.

\bibitem{Parsons78Pursu}
Torrence~D. Parsons.
\newblock Pursuit-evasion in a graph.
\newblock In {\em Theory and applications of graphs}, volume 642 of {\em
  Lecture Nots in Mathematics}, pages 426--441, 1978.
\newblock \href {https://doi.org/10.1007/BFb0070400}
  {\path{doi:10.1007/BFb0070400}}.

\bibitem{RobsertsonS84GMIII}
Niel Roberston and Paul~D. Seymour.
\newblock Graph minors. {III}. planar tree-width.
\newblock {\em Journal of Combinatorial Theory, Series B}, 36:49--64, 1984.
\newblock \href {https://doi.org/10.1016/0095-8956(84)90013-3}
  {\path{doi:10.1016/0095-8956(84)90013-3}}.

\bibitem{SeymourT93Graph}
Paul~D. Seymour and Robin Thomas.
\newblock Graph searching and a min-max theorem for tree-width.
\newblock {\em Journal of Combinatorial Theory, Series B}, 58(1):22--33, 1993.
\newblock \href {https://doi.org/10.1006/jctb.1993.1027}
  {\path{doi:10.1006/jctb.1993.1027}}.

\bibitem{TakahashiUK94minim}
Atsushi Takahashi, Shuichi Ueno, and Yoji Kajitani.
\newblock Minimal acyclic forbidden minors for the family of graphs with
  bounded path-width.
\newblock {\em Discret. Math.}, 127(1-3):293--304, 1994.
\newblock \href {https://doi.org/10.1016/0012-365X(94)90092-2}
  {\path{doi:10.1016/0012-365X(94)90092-2}}.

\bibitem{TakahashiUK95Mixed}
Atsushi Takahashi, Shuichi Ueno, and Yoji Kajitani.
\newblock Mixed searching and proper-path-width.
\newblock {\em Theoretical Computer Science}, 137(2):253--268, 1995.
\newblock \href {https://doi.org/10.1016/0304-3975(94)00160-K}
  {\path{doi:10.1016/0304-3975(94)00160-K}}.

\bibitem{ColinDeVerdiere98multi}
Yves Colin~De Verdi\`{e}re.
\newblock Multiplicities of eigenvalues and tree-width of graphs.
\newblock {\em Journal of Combinatorial Theory, Series B}, 74(2):121--146,
  1998.
\newblock \href {https://doi.org/10.1006/jctb.1998.1834}
  {\path{doi:10.1006/jctb.1998.1834}}.

\end{thebibliography}

%---------------------------------------------------------------------------------------------------------------
\appendix

%\newpage
%\input{SECTIONS/old_stuff}

\end{document}